\documentclass[11ptd,a4paper]{article}
\usepackage[left=2cm, right=2cm, top=2cm, bottom=2cm]{geometry}

\usepackage{hyperref}
\usepackage{soul}
\usepackage[T1]{fontenc}
\usepackage[utf8]{inputenc}
\usepackage{titlesec}

\usepackage{amsmath,amsfonts,amssymb,amsthm}

\newtheoremstyle{nonitalic}
{3pt}
{3pt}
{}
{}
{\bfseries}
{.}
{.5em}
{}

\theoremstyle{plain}
\newtheorem{theorem}{Theorem}[section] 

\theoremstyle{nonitalic}
\newtheorem{definition}[theorem]{Definition}
\newtheorem{example}[theorem]{Example}
\newtheorem{remark}[theorem]{Remark}

\theoremstyle{plain}

\newtheorem{lemma}[theorem]{Lemma}
\newtheorem{proposition}[theorem]{Proposition}

\theoremstyle{plain} 
\newtheorem*{theorem*}{Theorem}

\usepackage{amstext} 
\usepackage{mathtools}

\usepackage{mathrsfs}
\usepackage{verbatim}
\usepackage{tikz}

\usepackage{graphicx}
\usepackage{letltxmacro}
\usepackage{subcaption}
\usepackage{float}	
\usepackage{parskip}
\usepackage[backend=biber, style=numeric,maxbibnames=99]{biblatex}
\DeclareFieldFormat{url}{}
\DeclareFieldFormat{doi}{\href{https://doi.org/#1}{#1}}

\renewbibmacro*{url+urldate}{%
  \iffieldundef{doi}{
    \printfield{url}
    \iffieldundef{urldate}{}
      {\setunit*{\addspace}%
       \printurldate}
  }{}
}

\renewbibmacro*{doi+eprint+url}{%
  \printfield{doi}
}

\usepackage[ruled]{algorithm2e}

\DontPrintSemicolon
\SetKwInput{Input}{input}
\SetKwInput{Output}{output}
\SetArgSty{textrm}

\SetFuncSty{MyFuncSty}
\SetAlCapNameSty{MyFuncSty}

\SetKwSty{MyKwSty}

\SetCommentSty{MyCommentSty}

\makeatletter
\let\original@algocf@latexcaption\algocf@latexcaption
\long\def\algocf@latexcaption#1[#2]{%
  \@ifundefined{NR@gettitle}{%
    \def\@currentlabelname{#2}%
  }{%
    \NR@gettitle{#2}%
  }%
  \original@algocf@latexcaption{#1}[{#2}]%
}
\makeatother

\newcommand*{\refalg}[1]{\FuncSty{\getrefbykeydefault{#1}{name}{TODO}} (\textrm{Algorithm}\ \ref{#1})}

\usepackage{ifthen}
\usepackage{color}             
\usepackage{tikz}    
\usepackage{tikz-3dplot}
\usepackage{hyperref}
\usepackage{xurl}
\hypersetup{breaklinks=true}

\usepackage{stmaryrd}
\usepackage{paralist}
\usepackage{esint}

\titleformat{\paragraph}
{\normalfont\normalsize\bfseries}{\theparagraph}{1em}{}
\titlespacing*{\paragraph}
{0pt}{3.25ex plus 1ex minus .2ex}{1.5ex plus .2ex}

\newcommand{\scalp}[2]{\langle #1,#2\rangle }

\makeatletter
\let\oldr@@t\r@@t
\def\r@@t#1#2{%
\setbox0=\hbox{$\oldr@@t#1{#2\,}$}\dimen0=\ht0
\advance\dimen0-0.2\ht0
\setbox2=\hbox{\vrule height\ht0 depth -\dimen0}%
{\box0\lower0.4pt\box2}}
\LetLtxMacro{\oldsqrt}{\sqrt}
\renewcommand*{\sqrt}[2][\ ]{\oldsqrt[#1]{#2}}
\makeatother

\usepackage[colorinlistoftodos]{todonotes}
\title{A Framework for Symmetric Self-Intersecting Surfaces}
\author{Christian Amend$^*$ and Tom Goertzen$^*$}
\date{}

\begin{document}

\maketitle

\def\thefootnote{*}\footnotetext{The authors contributed equally to this work. Emails: \texttt{christian.amend@rwth-aachen.de}, \texttt{tom.goertzen@rwth-aachen.de}}

\begin{abstract}
3D printing of surfaces has become an established method for prototyping and visualisation.  However, surfaces often contain certain degenerations, such as self-intersecting faces or non-manifold parts, which pose problems in obtaining a 3D printable file. Therefore, it is necessary to examine these degenerations beforehand. \\
Surfaces in three-dimensional space can be represented as embedded simplicial complexes describing a triangulation of the surface. We use this combinatorial description, and the notion of embedded simplicial surfaces (which can be understood as well-behaved surfaces) to give a framework for obtaining 3D printable files. This provides a new perspective on self-intersecting triangulated surfaces in three-dimensional space.
Our method first retriangulates a surface using a minimal number of triangles, then computes its outer hull, and finally treats non-manifold parts. To this end, we prove an initialisation criterion for the computation of the outer hull. We also show how symmetry properties can be used to simplify computations. Implementations of the proposed algorithms are given in the computer algebra system GAP4. To verify our methods, we use a dataset of self-intersecting symmetric icosahedra.
Exploiting the symmetry of the underlying embedded complex leads to a notable speed-up and enhanced numerical robustness when computing a retriangulation, compared to methods that do not take advantage of symmetry.
\end{abstract}

\section{Introduction}

Triangulated surfaces in three-dimensional space are an essential tool in 3D printing and in geometric modelling. These surfaces are usually described by their incidence structure, consisting of vertices, edges and faces and coordinates for each vertex. For applications like 3D printing, certain regularity properties are often assumed, such as the absence of \emph{self-intersecting faces} or \emph{non-manifold parts}, as these lead to artefacts in the printed models. These degenerations frequently appear in surfaces, and it is thus necessary to treat them before obtaining a 3D printable file.  In this context, non-manifold parts are either \emph{non-manifold edges} that are incident to more than two faces or \emph{non-manifold vertices}, where the incident faces cannot be ordered in a connected face-edge path. Non-manifold edges in particular cause parts to be disconnected when printing. For instance, when identifying two cubes at an edge, we obtain a surface with a non-manifold edge as shown in Figure \ref{fig:TouchingCubes}. Common \emph{slicer software} that prepares a file for 3D printing often neglects these edges, leading to 3D printed cubes that no longer share a common edge. 
Self-intersecting faces yield another problem, as they hinder the computation of the outer-hull, which is necessary to determine a printing path. This is because inner parts of the model are usually disregarded in the printing process. In Figure \ref{fig:IntersectingCubes} an example of this with two intersecting cubes is shown.

\begin{figure}[H]
\centering
\begin{minipage}{0.49\textwidth}
   \begin{subfigure}{\textwidth}
   \centering
    \includegraphics[height=3cm]{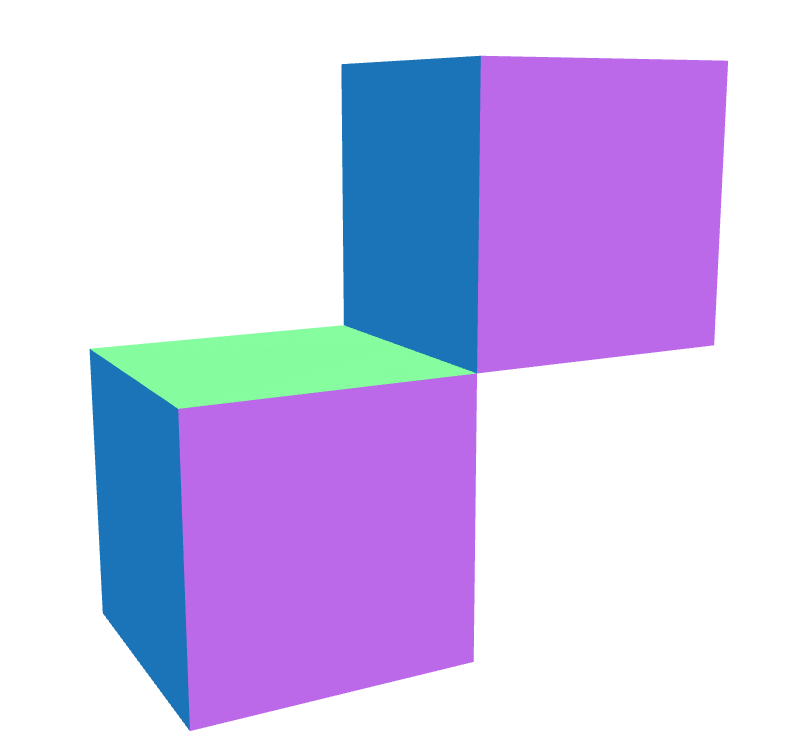}
    \includegraphics[height=3cm]{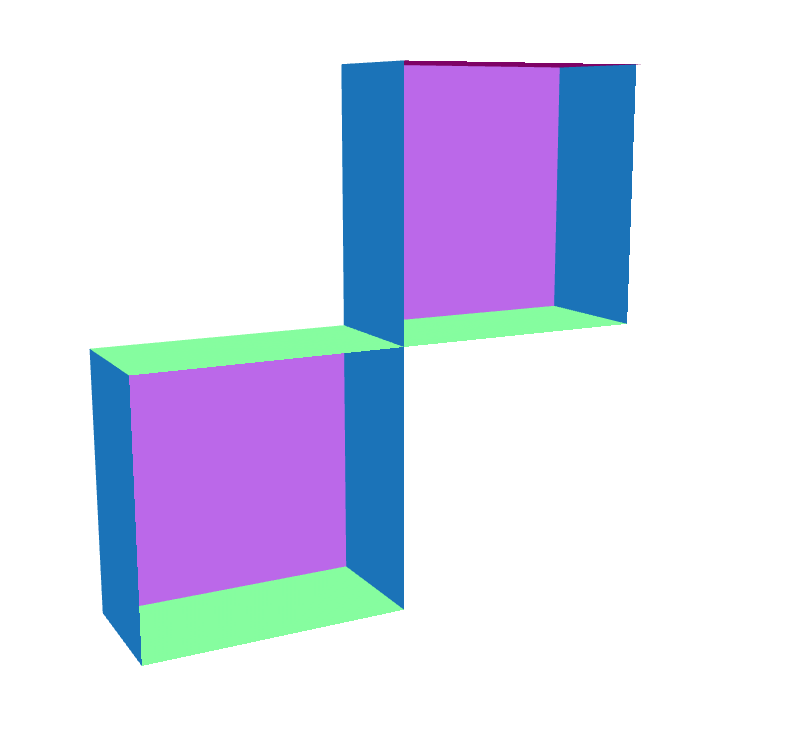}
    \caption{}
    \label{fig:TouchingCubes}
   \end{subfigure} 
\end{minipage}
\begin{minipage}{0.49\textwidth}
   \begin{subfigure}{\textwidth}
   \centering
    \includegraphics[height=3cm]{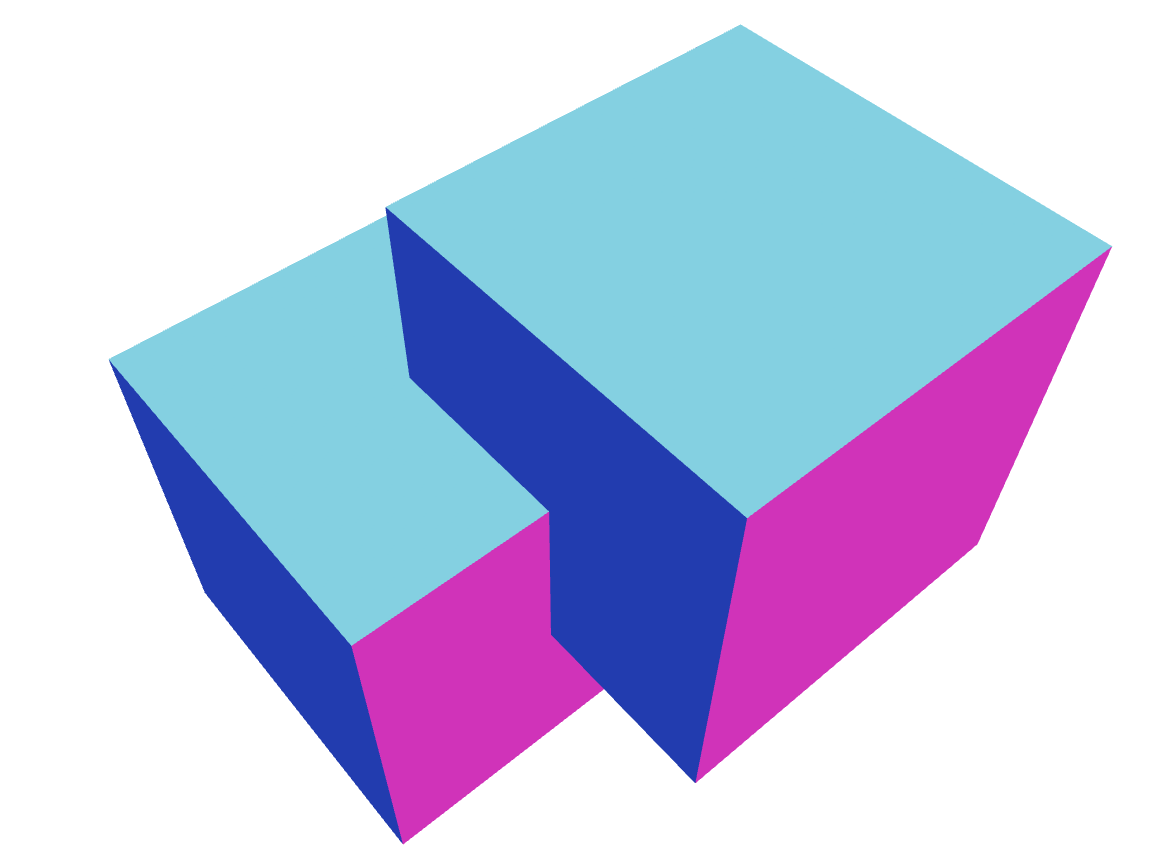}
    \includegraphics[height=3cm]{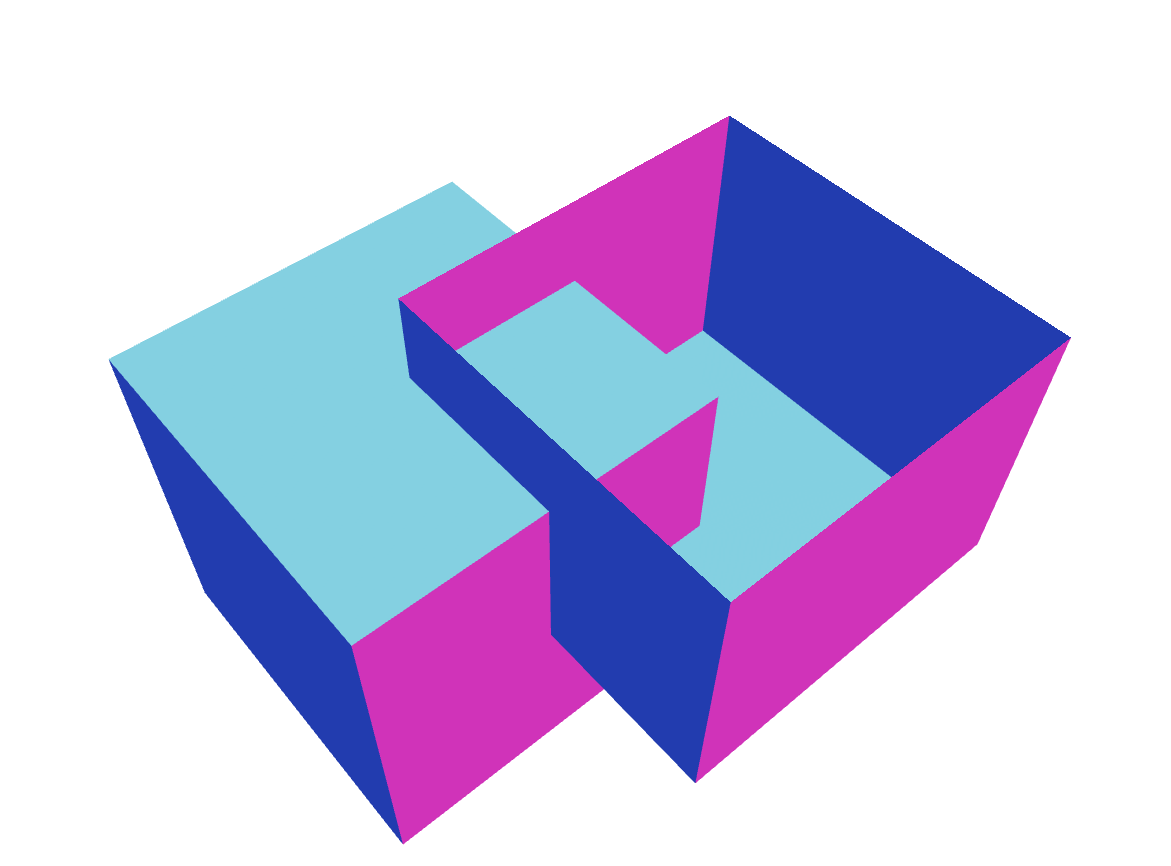}
    \caption{}
    \label{fig:IntersectingCubes}
   \end{subfigure} 
\end{minipage}
\caption{(a) Identifying two cubes at an edge leads to a non-manifold edge. In 3D printing applications, this edge is often neglected, leading to two separated cubes. (b) Two intersecting cubes are shown together with a view of its interior. The interior part can be omitted by reduction to the outer-hull.}
\label{fig:cube_examples}
\end{figure}
In order to obtain 3D printable files that can be realised as models depicting the key features of the original surface, it is necessary to address the aforementioned problems. In Figure \ref{fig:cube_examples_fixed}, two modified versions of these cubes are shown that are geometrically as close as possible (for instance using the \emph{Hausdorff distance}) such that 3D printed copies do not display artefacts. Additionally, the underlying surface is well-behaved in the context above, i.e.\ does not possess self-intersections or non-manifold parts.

\begin{figure}[H]
\centering
\begin{minipage}{0.49\textwidth}
   \begin{subfigure}{\textwidth}
   \centering
    \includegraphics[height=3cm]{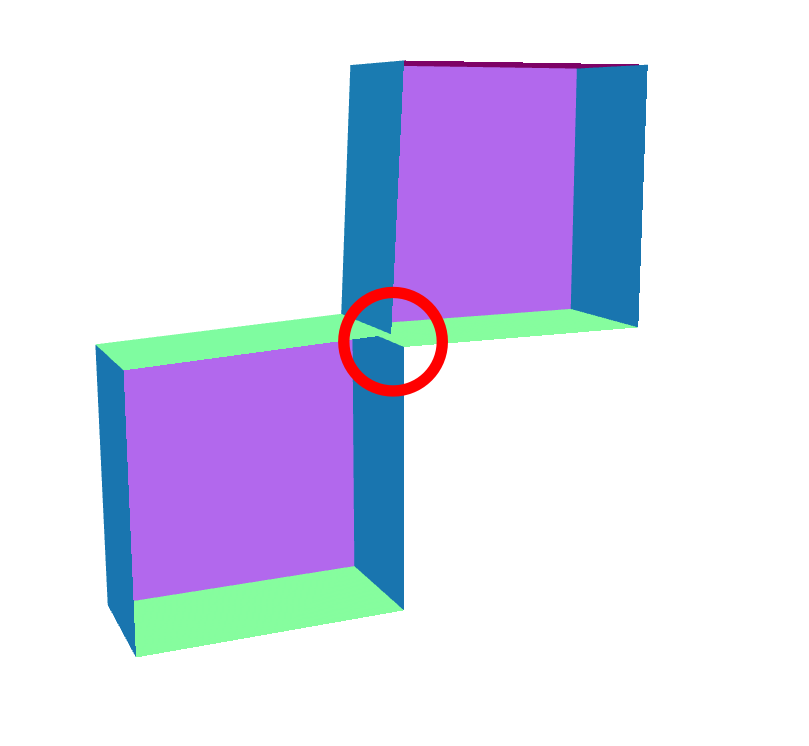}
    \includegraphics[height=3cm]{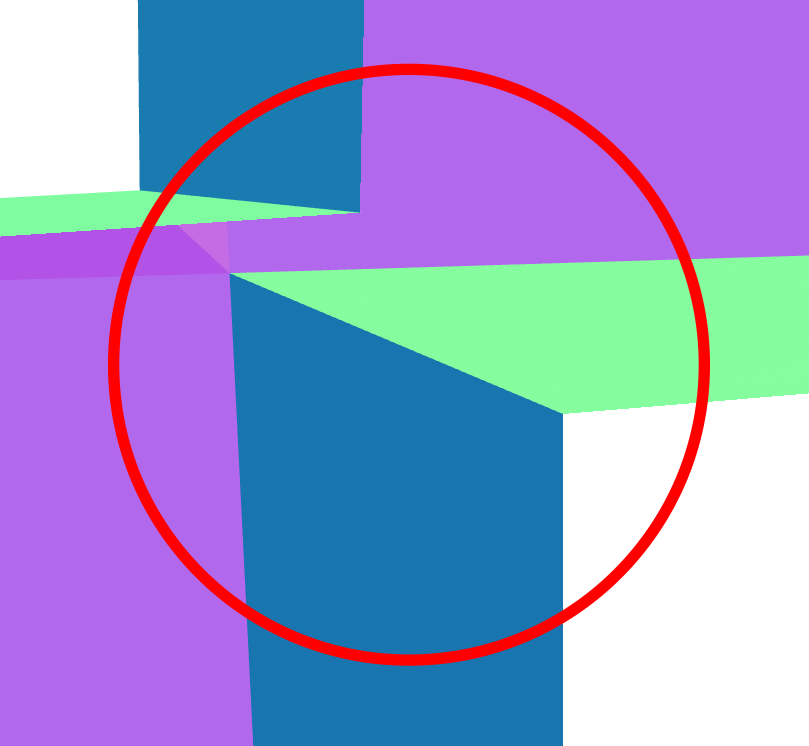}
    \caption{}
    \label{fig:TouchingCubesFixed}
   \end{subfigure} 
\end{minipage}
\begin{minipage}{0.49\textwidth}
   \begin{subfigure}{\textwidth}
   \centering
    \includegraphics[height=3cm]{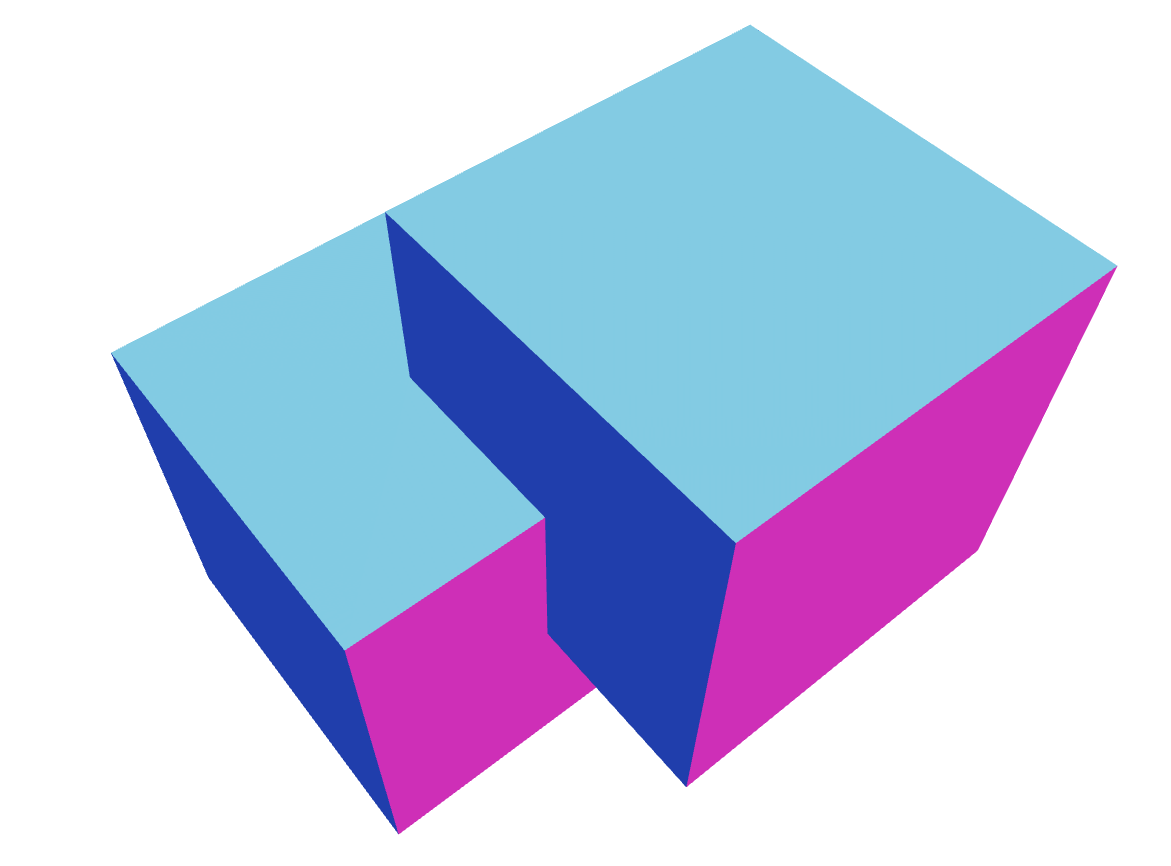}
    \includegraphics[height=3cm]{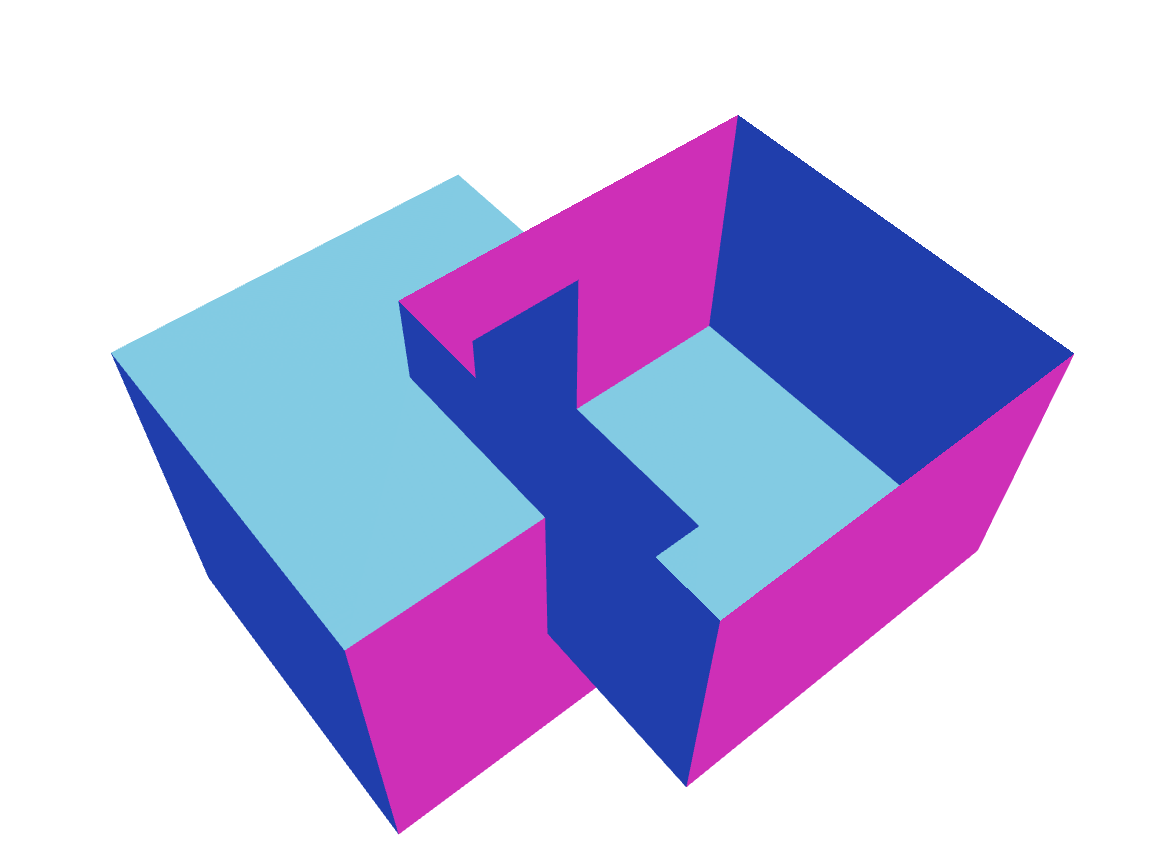}
    \caption{}
    \label{fig:IntersectingCubesFixed}
   \end{subfigure} 
\end{minipage}
\caption{(a) We can modify the two cubes in Figure \ref{fig:TouchingCubes} by operations replacing the part containing the non-manifold edge. When 3D printing the adapted surface, the two cubes are connected. (b) For 3D printing applications, only the outer-hull is printed. Thus, we obtain a modified surface by neglecting inner parts.}
\label{fig:cube_examples_fixed}
\end{figure}

In this paper, we address these challenges by introducing a framework that modifies surfaces to obtain 3D printable files. To this end, we first introduce the notion of \emph{(embedded) simplicial surfaces} in Section \ref{sec:preliminaries}, which are triangulated surfaces without non-manifold parts. This notion is based on the work presented in \cite{Brakhage,simplicialsurfacesbook}. Next, we describe our method, which starts with a triangulated surface, containing self-intersections or non-manifold parts and yields an embedded simplicial surface. Here, we focus on methods allowing the exploitation of symmetries of the model, leading to a robust way of detecting and rectifying self-intersections, see Section \ref{sec:symmetry}. The guiding examples we consider in this work are the $35$ symmetric icosahedra of edge length $1$ classified in \cite{IcosahedraEdgeLength1}.

Our approach proceeds as follows:
\begin{enumerate}
    \item Compute all self-intersections of the model in Section \ref{sec:detection}.
    \item Retriangulate the original faces such that the resulting complex has no self-intersections in Section  \ref{sec:rectif}.
    \item Compute the outer hull, chambers and correctly (outward) oriented normals of the retriangulated complex in Section \ref{sec:outer_hull}.
    \item Remedy non-manifold parts in Section \ref{sec:ramifications}.
\end{enumerate}

Note that the order in which these steps are performed is essential for our approach: for the outer hull, all self-intersections have to be removed, and for tackling non-manifold parts, one needs to reduce to the outer hull with outward oriented normals. More details on this can be found in the respective sections.
In the following, we highlight steps 1.-3. with the great icosahedron. In this case, step 4 is not necessary, as no non-manifold parts are present.

The implementation of our methods in GAP4 \cite{GAP4} are available in our GAP4 package \cite{Amend_Goertzen_2024}. 

\paragraph*{The Great Icosahedron}

A prominent example of a self-intersecting polyhedron is the great icosahedron, shown in Figure \ref{fig:great_icosahedron}.   The great icosahedron is a non-convex regular polyhedron with icosahedral symmetry and part of the Kepler-Poinsot polyhedra, see \cite{coxeter1982fiftynine}. Like the regular icosahedron, a Platonic solid, it has $20$ faces, $30$ edges and $12$ vertices and thus Euler characteristic $2$. Moreover, it is invariant under the full icosahedral group of order $120$. However, it has many self-intersections: each of the $20$ faces intersects with $15$ faces non-trivially. For processing this surface for applications like 3D printing or meshing, it is natural to ask how the symmetry of the given surface can be exploited to simplify our preprocessing. Specifically, we can investigate how symmetry simplifies the computation of the outer hull or that of a retriangulation. For the computation of a new triangulation that corresponds to the embedding of the great icosahedron in $\mathbb{R}^3$, we can proceed as follows: we note that each face can be mapped to any other face by at least one of the $120$ symmetries, i.e.\ the symmetry group acts transitively on the faces. It follows that it suffices to first retriangulate one face, and in a second step transfer this retriangulation to the entire surface.  Similarly, we can use the symmetry group when computing self-intersections of face pairs. Without the use of symmetry, one would need to consider all $\binom{20}{2}=190$ face pairs when searching for self-intersections. In the case of the great icosahedron, symmetry group acting on pairs of faces has $5$ orbits, meaning it suffices to consider $5$ face pairs instead of all $190$ to compute all self-intersections. Using an algorithm that guarantees to consider a minimal number of triangles, see Section \ref{sec:rectif}, we retriangulate one triangle, and transfer this triangulation to all other triangles.  This is shown in Figure \ref{fig:great_icosahedron_face_1}. Next, we can compute the outer hull to obtain a surface with $180$ faces, $92$ vertices and $270$ edges with icosahedral symmetry.

\begin{figure}[H]
\centering
\begin{minipage}{0.3\textwidth}
   \begin{subfigure}{\textwidth}
   \centering
    \includegraphics[height=5cm]{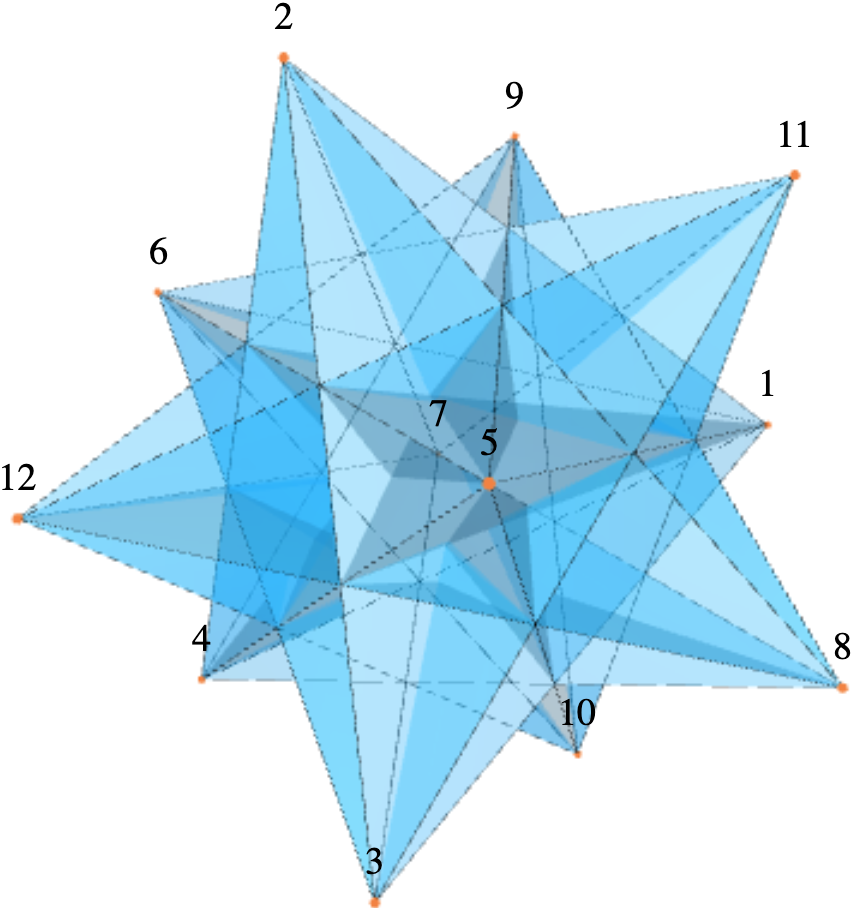}
    \caption{}
   \end{subfigure} 
\end{minipage}
\begin{minipage}{0.3\textwidth}
   \begin{subfigure}{\textwidth}
   \centering
    \includegraphics[height=5cm]{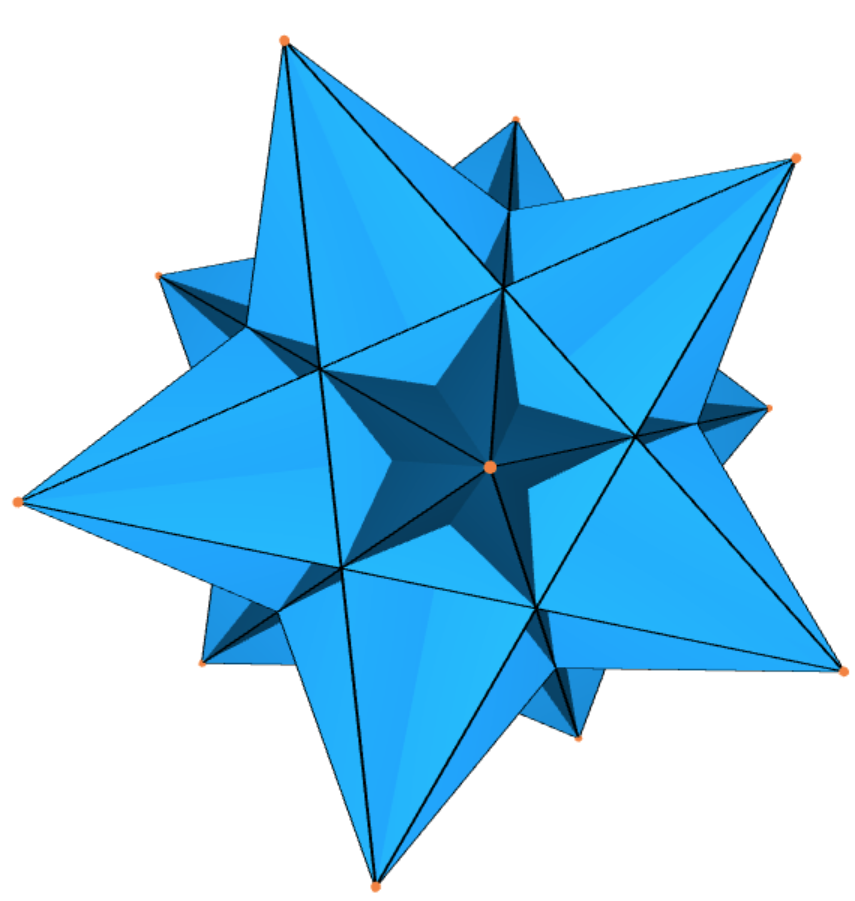}
    \caption{}
   \end{subfigure} 
\end{minipage}
\begin{minipage}{0.3\textwidth}
   \begin{subfigure}{\textwidth}
   \centering
    \includegraphics[height=5cm]{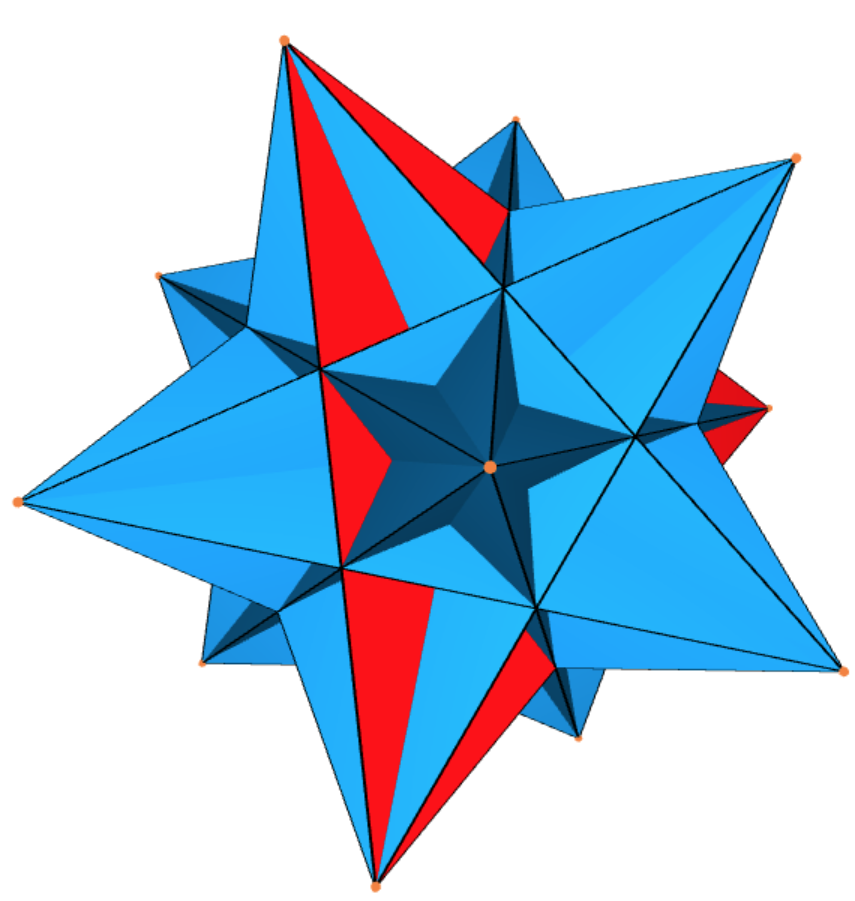}
    \caption{}
   \end{subfigure} 
\end{minipage}
\caption{(a) The great icosahedron has 12 embedded vertices, 30 edges and 20 faces, and the same incidence structure as the regular icosahedron (a Platonic solid). (b,c) Each of the 20 faces intersects with 15 faces non-trivially, and we can compute the intersections of all face pairs by only considering one face (face in red) and in a next step use the symmetry group to compute the remaining intersections.}
\label{fig:great_icosahedron}
\end{figure}

To be more specific, we can choose the coordinates of the vertices of the great icosahedron in such a way that all edge lengths are of unit length:
\begin{align*}
 & \overset{v_{1}}{\left(0,\frac{1}{2\varphi},\frac{1}{2}\right)},\overset{v_{2}}{\left(0,\frac{1}{2\varphi},-\frac{1}{2}\right)},\overset{v_{3}}{\left(-\frac{1}{2\varphi},-\frac{1}{2},0\right)},\overset{v_{4}}{\left(\frac{1}{2\varphi},-\frac{1}{2},0\right)},\overset{v_{5}}{\left(-\frac{1}{2},0,-\frac{1}{2\varphi}\right)},\overset{v_{6}}{\left(\frac{1}{2},0,-\frac{1}{2\varphi}\right)},\\
 & \overset{v_{7}}{\left(\frac{1}{2},0,\frac{1}{2\varphi}\right)},\overset{v_{8}}{\left(-\frac{1}{2},0,\frac{1}{2\varphi}\right)},\overset{v_{9}}{\left(\frac{1}{2\varphi},\frac{1}{2},0\right)},\overset{v_{10}}{\left(0,-\frac{1}{2\varphi},\frac{1}{2}\right)},\overset{v_{11}}{\left(-\frac{1}{2\varphi},\frac{1}{2},0\right)},\overset{v_{12}}{\left(0,-\frac{1}{2\varphi},-\frac{1}{2}\right)},
\end{align*}

where $\varphi=\frac{1+\sqrt{5}}{2}$ is the \emph{golden ratio}. The vertices of the faces are then given by the following list:

\[
\begin{aligned}
& \left[ v_1, v_2, v_3 \right], \left[ v_1, v_2, v_4 \right], \left[ v_1, v_4, v_5 \right], \left[ v_1, v_5, v_6 \right], \left[ v_1, v_3, v_6 \right], \left[ v_2, v_3, v_7 \right], \\
&\left[ v_2, v_4, v_8 \right], \left[ v_4, v_5, v_9 \right], \left[ v_5, v_6, v_{10} \right], \left[ v_3, v_6, v_{11} \right], \left[ v_2, v_7, v_8 \right], \left[ v_4, v_8, v_9 \right], \\
&\left[ v_5, v_9, v_{10} \right], \left[ v_6, v_{10}, v_{11} \right], \left[ v_3, v_7, v_{11} \right], \left[ v_7, v_8, v_{12} \right], \left[ v_8, v_9, v_{12} \right], \left[ v_9, v_{10}, v_{12} \right], \\
&\left[ v_{10}, v_{11}, v_{12} \right], \left[ v_7, v_{11}, v_{12} \right].
\end{aligned}
\]

For 3D printing purposes and other applications, it is essential to disregard inner parts of a model by computing the outer-hull. In order to tackle this algorithmically, we first identify all intersecting face pairs and disregard their parts that lie in the interior. In Figure \ref{fig:great_icosahedron_face_1}, we see the steps that are performed to compute the intersection of the face marked in red in Figure \ref{fig:great_icosahedron}.

We observe that the symmetry group of the great icosahedron is isomorphic to the full icosahedral group which itself is isomorphic to $C_2\times A_5$ (a direct product of the cyclic group of order $2$ with the alternating group of order $5$) and can be generated by the following three reflection matrices \[  \begin{pmatrix}
-1 & 0 & 0 \\
0 & 1 & 0 \\
0 & 0 & 1
\end{pmatrix} , \begin{pmatrix}
\frac{\varphi}{2} & -0.5 & \frac{1}{2\varphi} \\
-0.5 & -\frac{1}{2\varphi} & \frac{\varphi}{2} \\
\frac{1}{2\varphi} & \frac{\varphi}{2} & 0.5
\end{pmatrix},  \begin{pmatrix}
1 & 0 & 0 \\
0 & 1 & 0 \\
0 & 0 & -1
\end{pmatrix}. \]

As mentioned above, this group acts transitively on the faces. Thus, we can reduce the number of face pair intersections that need to be considered, as the stabiliser of a single face is isomorphic to a symmetry group of an equilateral triangle. This implies that we can compute all intersections using only one face, which is shown in Figure \ref{fig:great_icosahedron_face_1}.

\begin{figure}[H]
\centering
\begin{minipage}{0.24\textwidth}
   \begin{subfigure}{\textwidth}
   \centering
    \includegraphics[height=3.5cm]{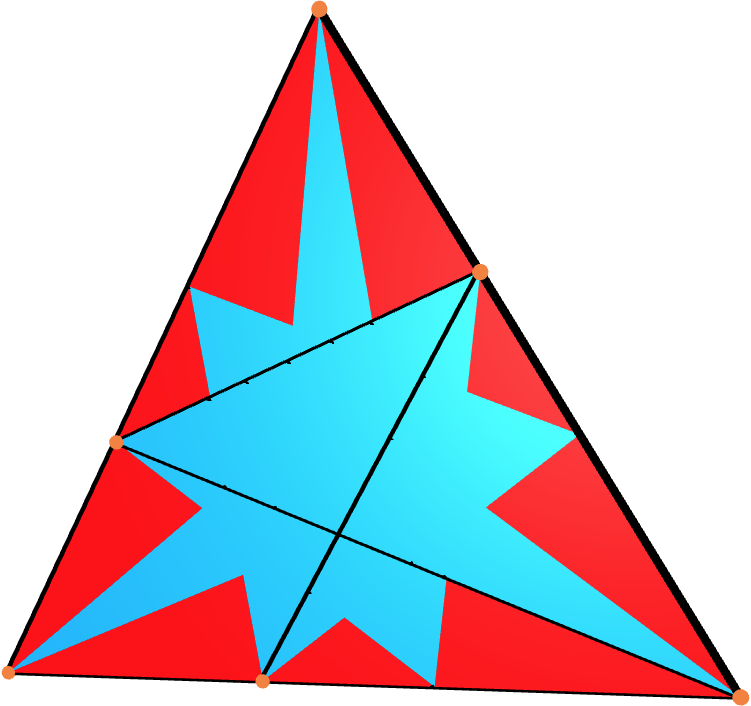}
    \caption{}
   \end{subfigure} 
\end{minipage}
\begin{minipage}{0.24\textwidth}
   \begin{subfigure}{\textwidth}
   \centering
    \includegraphics[height=3.5cm]{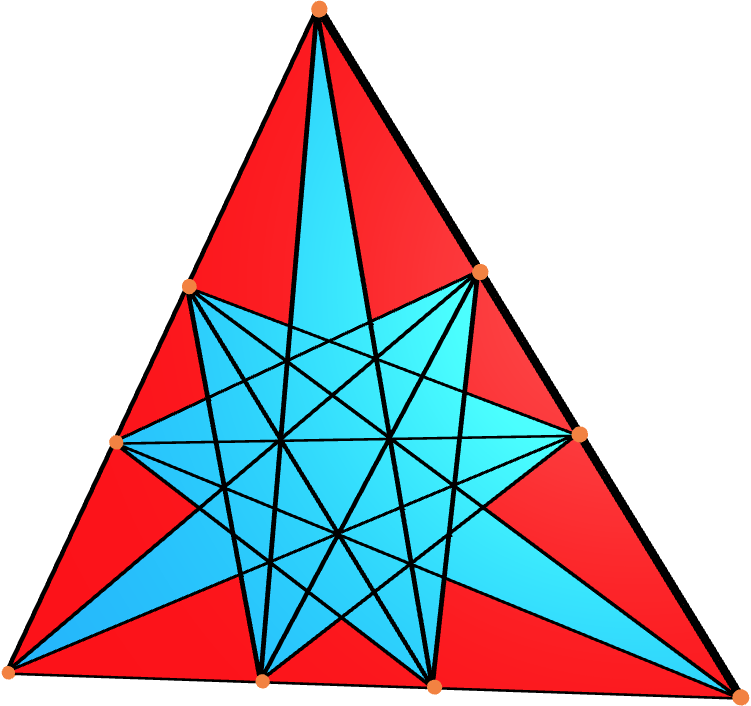}
    \caption{}
   \end{subfigure} 
\end{minipage}
\begin{minipage}{0.24\textwidth}
   \begin{subfigure}{\textwidth}
   \centering
    \includegraphics[height=3.5cm]{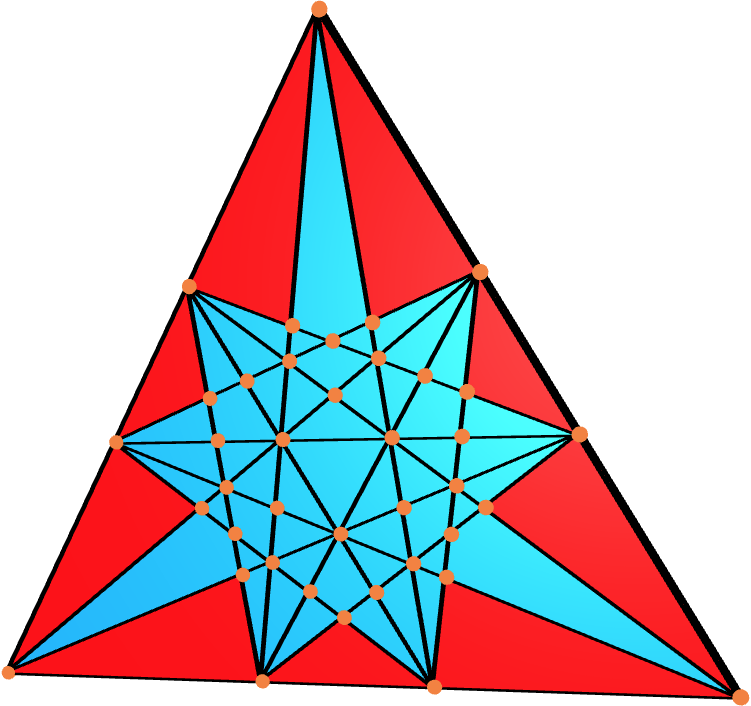}
    \caption{}
   \end{subfigure} 
\end{minipage}
\begin{minipage}{0.24\textwidth}
   \begin{subfigure}{\textwidth}
   \centering
    \includegraphics[height=3.5cm]{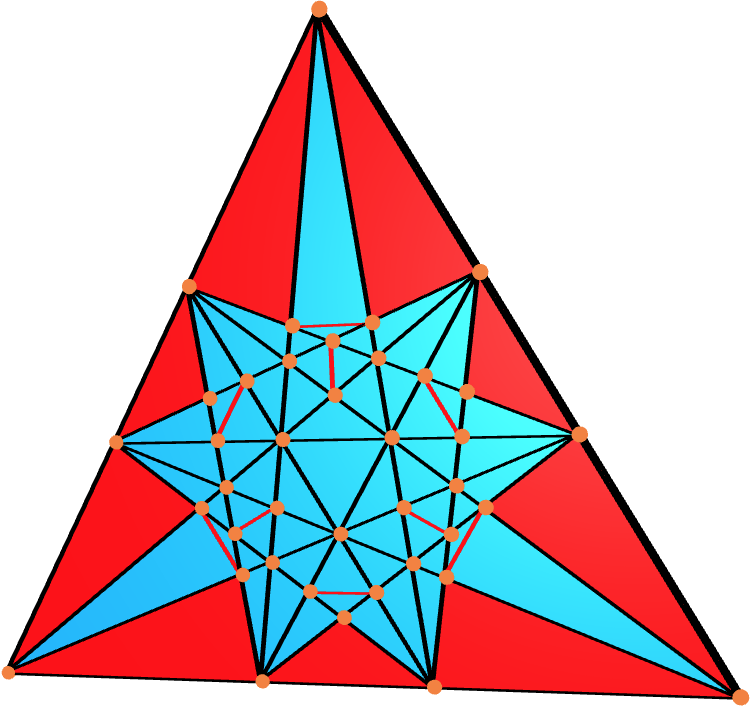}
    \caption{}
   \end{subfigure} 
\end{minipage}
\caption{(a) Intersections of one face of the great icosahedron with other faces up to symmetry. (b) The \emph{stellation diagram} of a face of the great icosahedron, showing lines where other face planes intersect with this one \cite{coxeter1982fiftynine}, can be obtained from (a) by rotations and reflections. (c) Computing all intersection points of intersection within the face. (d) Triangulating non-triangle parts, yielding a \emph{simplicial disc}. This retriangulation can be carried over to all 20 faces of the great icosahedron using its symmetry group in order to obtain a surface without self-intersections.}
\label{fig:great_icosahedron_face_1}
\end{figure}

For understanding the internal structure, we have to find the \emph{chambers} of the great icosahedron given by the connected components of the following set
$$\mathbb{R}^3 \setminus \bigcup_{f\in F}\mathrm{conv}(f),$$
where $\mathrm{conv}(f)$ is the \emph{convex hull} of the vertices defining $f$, see Definition \ref{def:outer_hull}. 
In Figure \ref{great_ico_explode}, we show an \emph{exploded view} of the 413 internal chambers of the great icosahedron obtained using the retriangulation as described above. Each chamber is shifted away from the centre with the same magnitude: let $p\in \mathbb{R}^3$ be the centre of the great icosahedron and for a given chamber $C$ with centre $c\in \mathbb{R}^3$, we shift the chamber $C$ using the translation $m\cdot (c-p)$ with magnitude $m\in \mathbb{R}_{>0}$.

\begin{figure}[H]
\begin{subfigure}{0.33\textwidth}
\includegraphics[height=5cm]{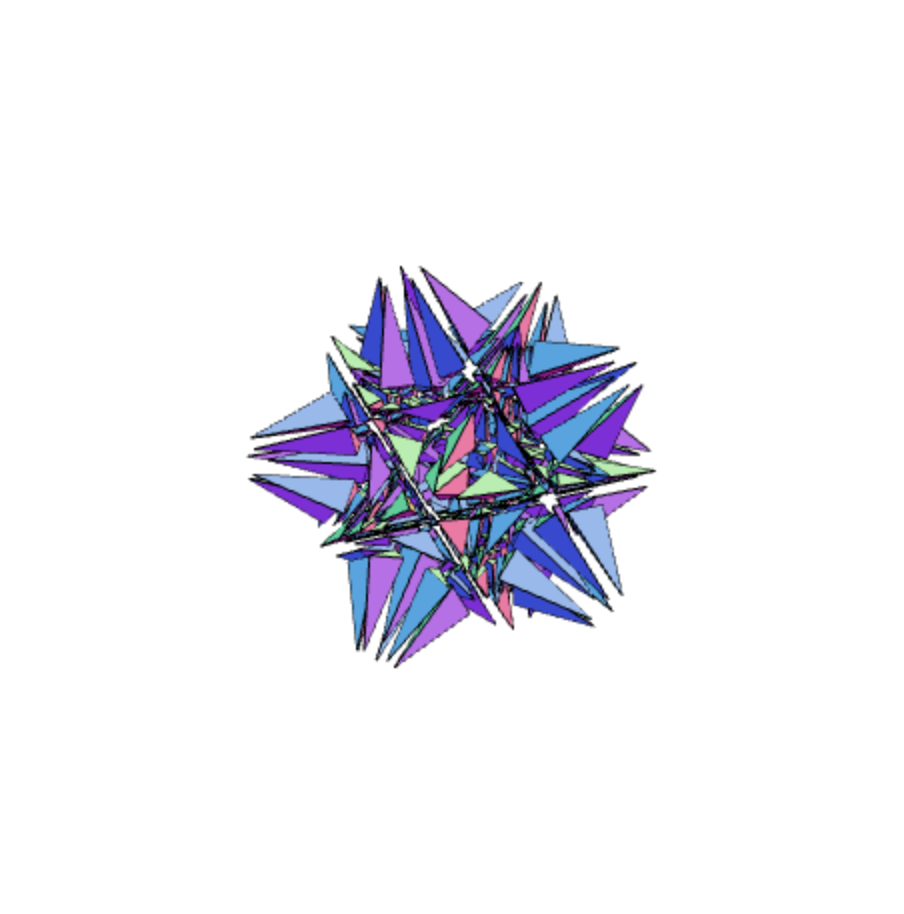}
\caption{$m=1$}
\end{subfigure}
\begin{subfigure}{0.33\textwidth}
\includegraphics[height=5cm]{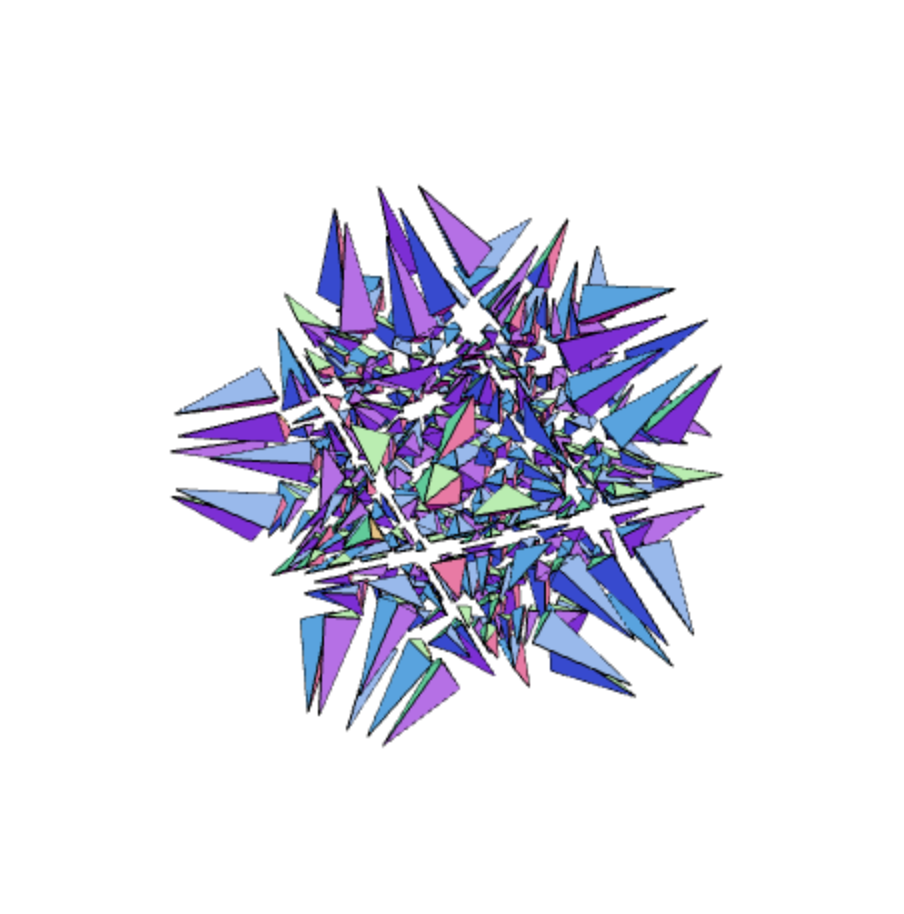}
\caption{$m=2$}
\end{subfigure}
\begin{subfigure}{0.33\textwidth}
\centering
\includegraphics[height=5cm]{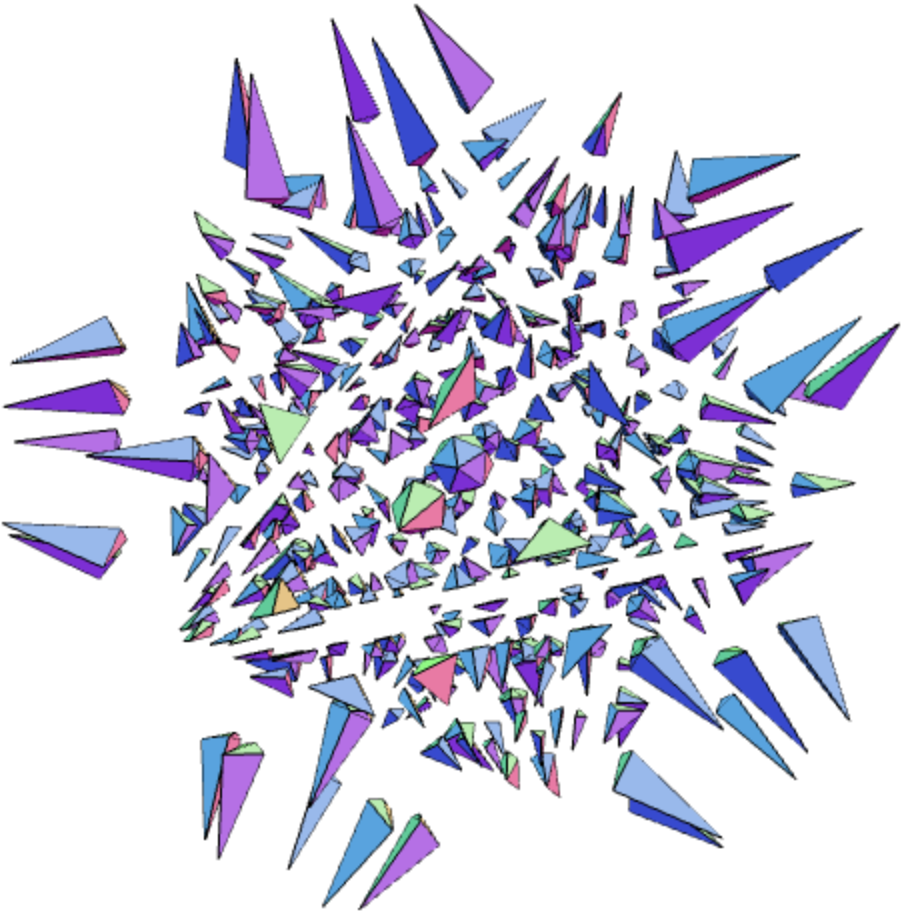}
\caption{$m=4$}
\end{subfigure}
\caption{Exploded views of the $413$ internal chambers of the great icosahedron with different magnitudes $m$.}
\label{great_ico_explode}
\end{figure}
The centre chamber of the great icosahedron with edge lengths $1$ is given by the regular icosahedron with faces corresponding to the central equilateral triangle of the stellation diagram in Figure \ref{fig:great_icosahedron_face_1} with edge lengths $$\left(\frac{2}{7+3\sqrt{5}}\right)\approx 0.145898.$$
In fact, all but one (final stellation) of the 59 icosahedra presented in \cite{coxeter1982fiftynine} can be obtained by taking a subset of the chambers that is invariant under the icosahedral group.

In this work, we mainly focus on another class of icosahedra, classified in \cite{IcosahedraEdgeLength1}, which contain all icosahedra that are combinatorially equivalent to the regular icosahedron, carry a non-trivial symmetry and have edge length $1$. There are exactly 35 such icosahedra, of which 33 possess self-intersections, and the regular icosahedron and the great icosahedron are also among them.

\paragraph*{Related Work}
Self-intersections and non-manifold parts of models are very active fields of research. This is because of the obstructions they cause in many fields, such as meshing, scanning of 3D models, and 3D printing \cite{Attene2018, RepairingMeshes2020, NonManifoldPointCloud, RobustAlgorithmSelfIntersections}.\\
In \cite{Attene2018}, an algorithm, along with an initialization criterion, to compute the outer hull of a complex is presented. In \cite{RobustAlgorithmSelfIntersections}, a robust method  is described to rectify self-intersections of a complex using a subdivision based on Delaunay triangulation with certain constraints to retriangulate the initial complex.\\
Focusing on self-intersections in triangulated complexes, there are several approaches for finding all self-intersections or testing whether two given triangles intersect, for instance \cite{TriangleTriangleIntersection,FastTriangleTriangleIntersectionMoeller}. Also, see \cite{Skala2023} for a recent review on algorithms for the detection of self-intersections. Repairing and retriangulation methods can be found in \cite{ATTENE_Repair_Meshes} for a direct approach, which is also suitable for the computation of outer hulls, \cite{ImmersionSelfIntersecting} for a method using immersion techniques and \cite{RemovingSelfIntersectionEdgeSwapping} for a method employing edge swap techniques. In \cite{RobustSelfIntersectionsCampenKobbelt}, a method is presented to change the topology of a polygonal mesh that combines an adaptive octree with nested binary space partitions (BSP), i.e.\  subdivisions of a Euclidean space into convex sets using hyperplanes as partitions. Alternative methods on remedying non-manifold parts can be found in \cite{VisualizeNonManifoldEdges,Non-manifold, NonManifoldPointCloud}. Another similar direction is the treatment of meshes and their repair, as done in \cite{RepairingMeshes2020}. As an input to our model, we use the symmetries of a given complex, which leads to simplification and speed up the algorithms. Of course, this is an idea that can be applied to different settings, such as in model segmentation in \cite{SymmetryModelSegmentation}. In conjunction with the ideas we present, one can consider detection of symmetries in surfaces, as in \cite{NewSymmetryDetection,EfficientSymmetry}.

\newpage

\section{Embedded Simplicial Surfaces} \label{sec:preliminaries}
In this section, we introduce the main terminology used in this work with a focus on \emph{embedded simplicial surfaces}, which yield well-behaved triangulations in the context of meshing and 3D printing applications. We start by giving a definition of a version of simplicial complexes adapted to a triangular surface, motivated by the combinatorial theory of simplicial surfaces, see \cite{BaumeisterPhDThesis,simplicialsurfacesbook}. In the literature, an (abstract) simplicial complex $X$ is commonly defined as a subset $X\subset \mathcal{P}(V)=\{\emptyset\neq A\subset V \}$ where $V$ is a set, and we have that for all $t\in X$ and for all $\emptyset\neq x\subset t$ it follows that $x\in X$. Below, we restrict ourselves to the case where the maximal elements of $X$ are sets of size $3$ and the elements of $X$ fulfil further conditions which are natural in the context of triangulations.

\begin{definition}[Simplicial Complex]
    \label{def:simplical_complex}
    Let $V$ be a finite set. A (closed) \textit{simplicial complex} $X$ with vertices $V$ is a subset of $\mathcal{P}_3(V) = \{A\subset V: |A|\leq 3\}$ such that the following conditions hold.
    \begin{enumerate}[(i)]
    \item For all $v\in V$, $\{v\} \in X$. Additionally, $\emptyset \not\in X$.
    \item For all $ t \in X$ and $\emptyset \neq x\subset t$, it follows that $x \in X$. 
    \item For each $t\in X$ with $|t| <3$, we can find $|t'|>|t|$ with $t \subset t'$ and $t'\in X$.
    \item For each $f\in X$ with $\lvert f \rvert = 3$ and any $v_1,v_2 \in f$, we can find $f' \neq f$ such that $v_1,v_2 \in f'$.
\end{enumerate} 
We call the three-element sets in $X$ the \textit{faces} or \textit{triangles}, the two-element sets in $X$ the \textit{edges}  and the one-element sets in $X$ the \textit{vertices}. The faces, edges and vertices are denoted by $X_2$, $X_1$ and $X_0$, respectively. We also say a vertex $v \in X$ is incident to an edge $e \in X$ if $v \subset e$, and an edge $e \in X$ is incident to a face $f \in X$ if $e \subset f$. Additionally, a vertex $v \in X$ is incident to a face $f \in X$ if $v \subset f$.
Since we do not consider complexes other than simplicial ones, we sometimes omit the word simplicial in the following. 
\end{definition}

The conditions in the definition above all correlate to natural assumptions for a simplicial complex consisting of triangles. For example, since we consider $\mathcal{P}_{3}(V)$, the faces are triangles. The conditions (i)-(iv) in the definition above are interpreted as follows:
\begin{enumerate}[(i)]
    \item ensures that the vertex set, on which a complex is built, is part of the complex.
    \item implies that if we take a face or edge, its parts are also included in the description of the complex.
    \item enforces that each vertex is part of an edge and each edge is part of a face.
    \item forces the surface to be closed, thus every edge has to be incident to at least two faces.
\end{enumerate}

In summary, we observe that a simplicial complex $X$ is uniquely determined by the incident vertices of its faces.

We give the following definition of faces associated to a given vertex.
\begin{definition}
    Let $X$ be a simplicial complex and take a vertex $v$ of $X$. Define the faces incident to $v$ to be the set $X_2(v)\coloneqq\{f \in X: v\in f \text{ and } |f|=3\}$.
\end{definition}
A more regular object is a simplicial surface, which enforces further properties compared to a simplicial complex.
\begin{definition}[Simplicial surface]
    \label{def:simplicial_surface}
    A (closed) \textit{simplicial surface} $S$ is a simplicial complex that additionally fulfils the following two conditions.
    \begin{enumerate}[(i)]
    \item Each edge $e\in S$ is incident to exactly two faces.
    \item For all vertices $v\in V$, we can order all faces $f_1,\dots,f_n$ incident to $v$ in a cycle, i.e.\ we can write $(f_1,\dots,f_n)$ such that $f_i$ and $f_{i+1}$ (and additionally $f_1$ and $f_n$) share a common edge.
\end{enumerate}
\end{definition}

The above conditions turn out to be crucial for identifying non-manifold parts of embedded simplicial complexes considered in the following sections. For example, the first condition directly corresponds to no \textit{non-manifold edges} being present in a given complex. The second condition is called the \textit{umbrella condition}. The vertices that do not fulfil this condition in a given complex are called \textit{non-manifold vertices} (see also Section \ref{sec:ramifications} for this). More general definitions of simplicial surfaces, also allowing distinct faces with the same set of incident vertices, can be found in \cite{akpanya2023surfaces,BaumeisterPhDThesis,simplicialsurfacesbook}.

\begin{definition}[Embedding]\label{def:embedding}
    An \textit{embedding} of a complex $X$ with vertices $V$ into $\mathbb{R}^3$ is an injective map 
    \begin{align*}
        \phi:V\to \mathbb{R}^3.
    \end{align*}
    We write $(X,\phi)$ for a complex $X$ with embedding $\phi$ and we refer to $(X,\phi)$ as an \emph{embedded complex}. 
    If a given complex $X$ is embedded into $\mathbb{R}^3$, we omit the map $\phi$ whenever it is clear from the context. Moreover, we often identify the vertices, edges and faces with their respective  images under $\phi$.
\end{definition}

In the definition above, the map $\phi$ is chosen to be injective to avoid degenerate edges and faces. Therefore, we can also obtain the initial underlying complex $X$ from its embedding $\phi(X)$.
In Figure \ref{fig:different_icos}, we see several distinct embeddings of the same underlying simplicial surface.

\begin{figure}[H]
\centering
\begin{minipage}{0.24\textwidth}
   \begin{subfigure}{\textwidth}
   \centering
    \includegraphics[height=3.5cm]{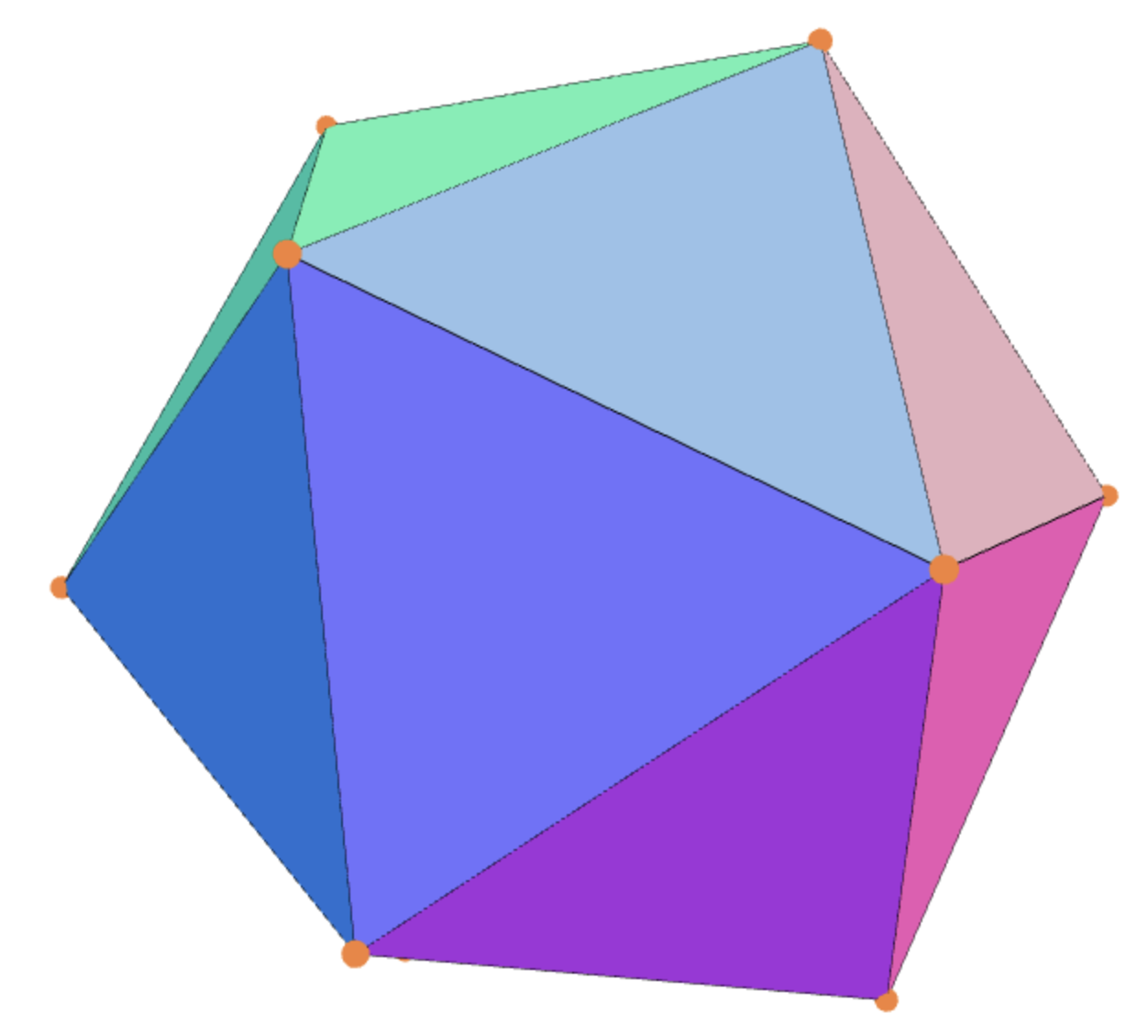}
    \caption{}
    \label{platonic_ico}
   \end{subfigure} 
\end{minipage}
\begin{minipage}{0.24\textwidth}
   \begin{subfigure}{\textwidth}
   \centering
    \includegraphics[height=3.5cm]{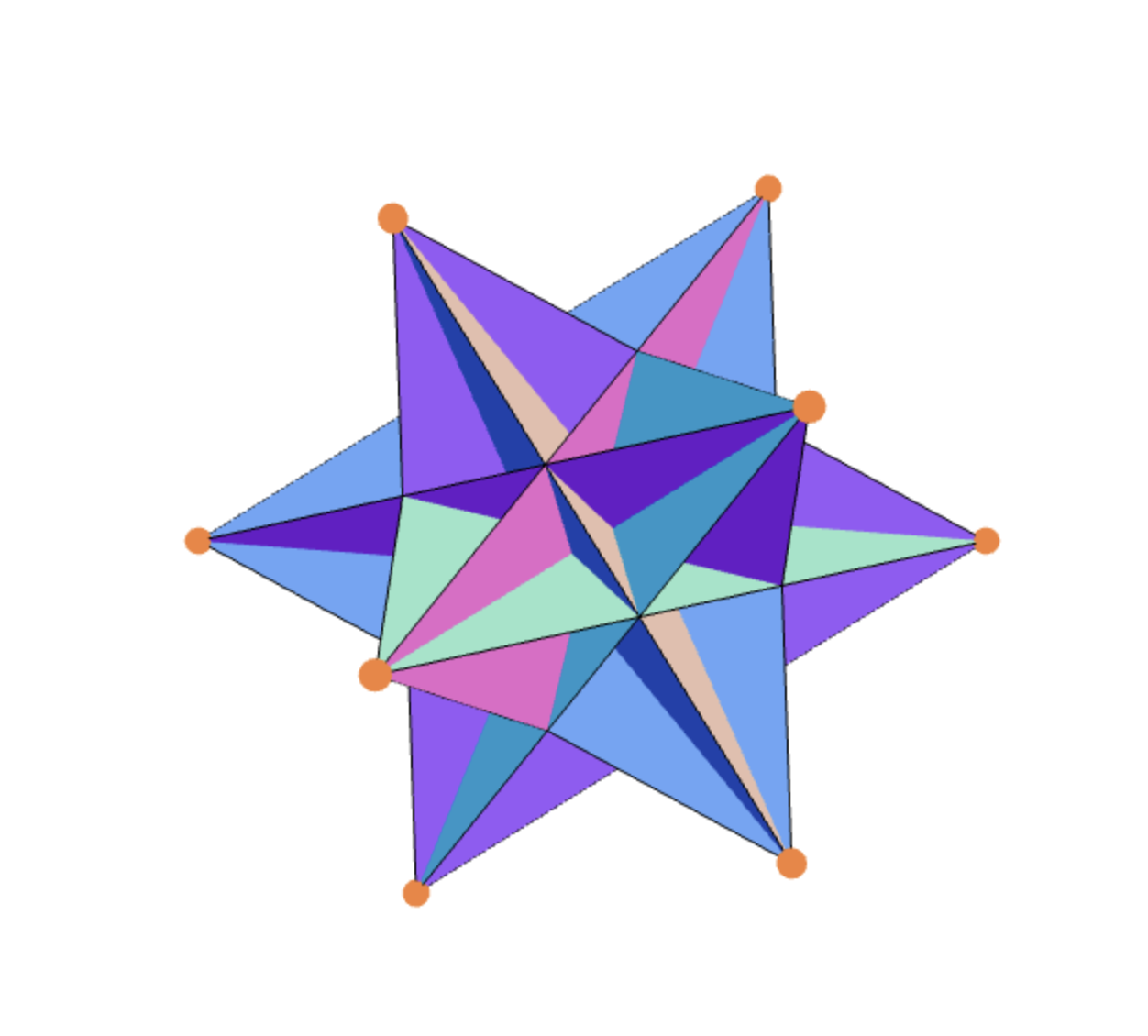}
    \caption{}
   \end{subfigure} 
\end{minipage}
\begin{minipage}{0.24\textwidth}
   \begin{subfigure}{\textwidth}
   \centering
    \includegraphics[height=3.5cm]{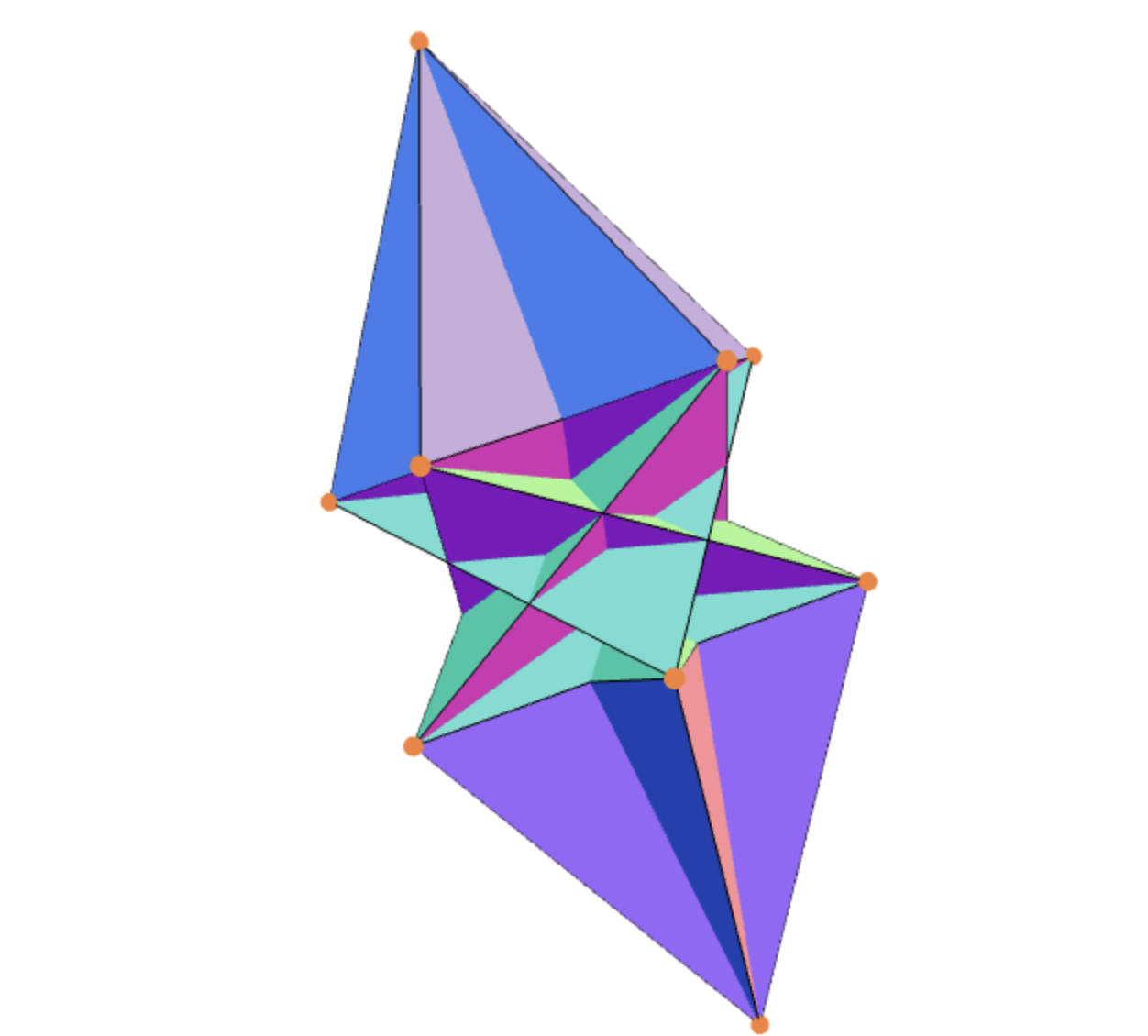}
    \caption{}
   \end{subfigure} 
\end{minipage}
\begin{minipage}{0.24\textwidth}
   \begin{subfigure}{\textwidth}
   \centering
    \includegraphics[height=3.5cm]{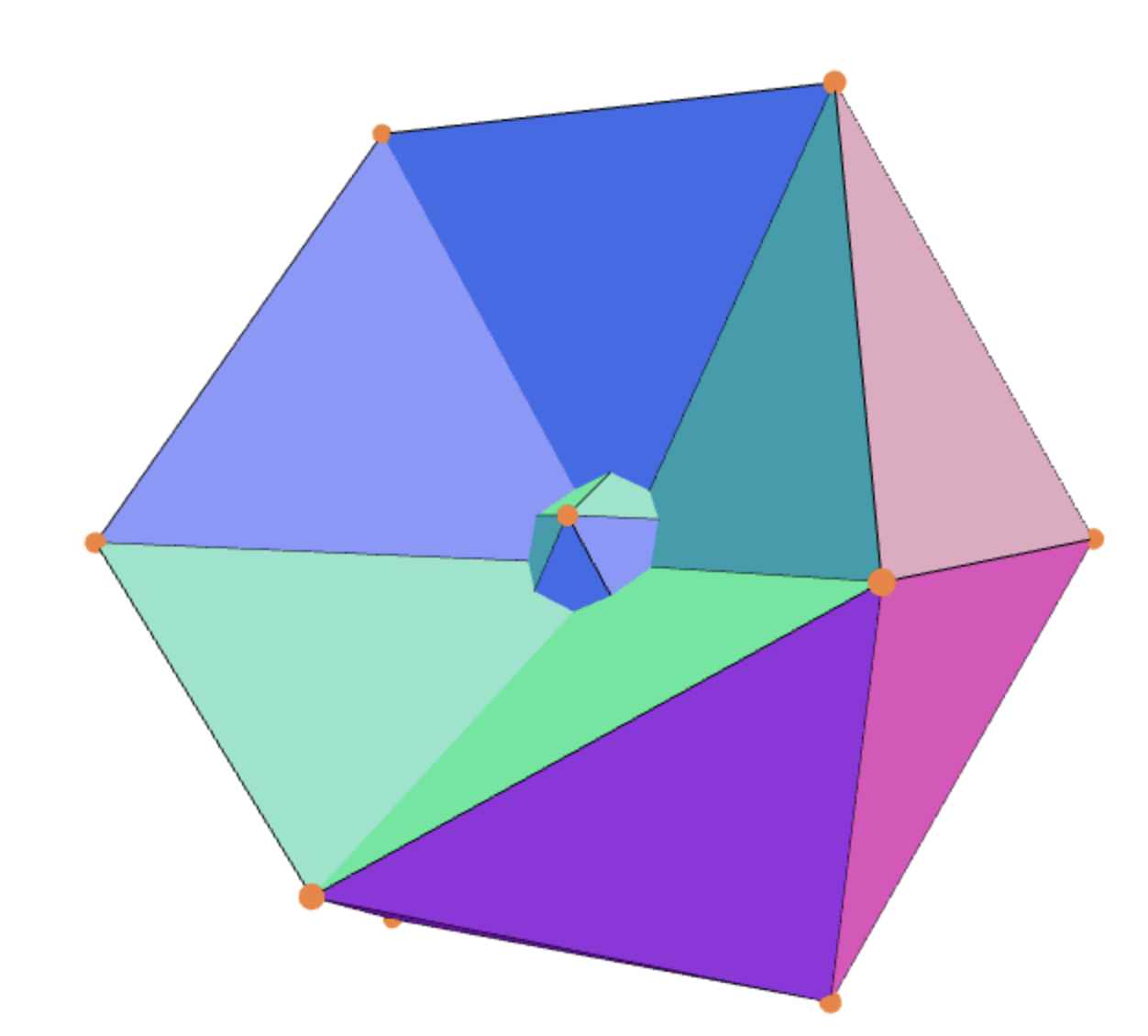}
    \caption{}
   \end{subfigure} 
\end{minipage}
\caption{One simplicial surface (icosahedron) with distinct embeddings of unit edge length \cite{IcosahedraEdgeLength1}: (a) platonic solid, (b) the great icosahedron, (c) icosahedron$_{3,1}$, (d) icosahedron$_{3,2}$ (notation given in \cite{Icosahedra:online}).}
\label{fig:different_icos}
\end{figure}

\begin{remark}
    The map $\phi$ in Definition \ref{def:embedding} can be represented as a list with $|V|$ entries in $\mathbb{R}^3$, or alternatively as a matrix in $\mathbb{R}^{|V| \times 3}$. Thus, one can switch from the combinatorial structure of an embedded simplicial complex to one used in application, such as a coordinate representation or an STL file. Here, it suffices to store the embedding data suitably, for example, via a list of lists.
\end{remark}

Motivated by the central objective of obtaining the outer hull of a surface, we define the orientation of a simplicial surface in a combinatorial way.

\begin{definition}[Orientation]\label{def:orientation}
    An \textit{orientation} of a simplicial surface $S$ is given by a cyclic ordering of the three vertices of each face, such that for any given edge $\{v_1,v_2 \}=e=f_1\cap f_2$ of two incident faces, we have that either the ordering of $f_1$ is $(v,v_1,v_2)$ and the one of $f_2$ is $(v',v_2,v_1)$ or the ordering of $f_1$ is $(v,v_1,v_2 )$ and $f_2$ is $(v',v_1,v_2)$.
\end{definition}

This definition can be used to associate a normal vector for each face of an orientable surface.

\begin{remark}\label{rem:oriented_normals}
    If an embedded surface is orientable, we can define the outer normals for each face as follows: Let $f=\{v_1,v_2,v_3\}$ such that the vertices are ordered as $(v_1,v_2,v_3)$, then the outer normal is given by the right-hand rule $$(v_2-v_1) \times (v_3-v_1).$$
    Note that if the outer hull of our surface is given as the boundary of an open set, it is well-known to be orientable (one can use the outer normals). In applications such as 3D printing, this is to be expected after removing artefacts. 
\end{remark}

As an alternative to the combinatorial approach of defining an orientation and related outer normals for an embedding, we give a geometric definition of chambers and the outer hull for an embedded complex as follows.

\begin{definition}[Chambers and outer hull]
    \label{def:outer_hull}
    The \emph{chambers} of an embedded complex $X$ are defined as the connected components of $$\mathbb{R}^3\setminus \bigcup_{f\in X_2}\text{conv}(f),$$ where $\text{conv}(f)$ is the convex hull given by the vertices of $f$, i.e.\ $$ \mathrm{conv}(f)=\{\lambda_1 \cdot v_1+\lambda_2\cdot v_2+\lambda_3\cdot v_2 \mid 0\leq \lambda_1,\lambda_2,\lambda_3 \leq 1, \sum_{i=1}^3 \lambda_i=1,f=\{v_1,v_2,v_3 \} \}.$$ The \emph{outer hull} $X_{\text{out}}$ is defined as the boundary of the unique chamber with infinite volume. The chambers of $X$ which are not $X_{\mathrm{out}}$ are called the \textit{finite} chambers.
    
\end{definition}

Note that since we require an embedded complex $X$ to be finite, the outer hull in Definition  \ref{def:outer_hull} exists and is non-empty. 

With this in mind, we can also define when a point is contained inside the embedded complex.

\begin{definition}\label{def:contained}
     Let $X$ be an embedded simplicial complex. We say that a point $p \in \mathbb{R}^3$ is \emph{contained} in $X$ if $p\in C$ for a finite chamber $C$ of $X$ or $p\in \text{conv}(f)$ for a face $f\in X$. More generally, we say that for two elements in the embedded complex $x,y\in X$ that $x$ is contained in $y$ if $\mathrm{conv}(x)\subset \mathrm{conv}(y)$. This naturally extends to points  $p\in \mathbb{R}^3$: $p$ is contained in an element $x\in X$ if $p\in \mathrm{conv}(x)$. 
\end{definition}

Below, we define \emph{intersection points}, which parameterise the intersection of two faces of an embedded complex.

\begin{definition}[Intersection points]\label{def:intersection_points}
    Let $(X,\phi)$ be an embedded simplicial complex. We say that $X$ is a \emph{self-intersecting complex} or has \emph{self-intersections} if there exist two faces $f_1\neq f_2$ and a point $p\in \mathbb{R}^3$  such that $$p\in \mathrm{conv}(\phi(f_1))\cap \mathrm{conv}(\phi(f_2)) \setminus (f_1 \cap f_2),$$i.e.\ $p$ is not a common vertex of $f_1$ and $f_2$ and lies inside both faces. 
    Since the intersection of convex sets is again convex, we have that $\mathrm{conv}(\phi(f_1))\cap \mathrm{conv}(\phi(f_2)) \setminus (f_1 \cap f_2)$ can be written as the convex hull of finitely many points, called \emph{intersections points}, contained in the edges of $f_1$ and $f_2$. We write $I(X,\phi)$ for the collection of all intersection points and $I(X,\phi,f) := I(X,\phi) \cap \mathrm{conv}(\phi(f))$ for the collection of all intersection points of a fixed face $f$.   If the set of intersection points $I(X,\phi)$ is empty, we call the embedded complex \emph{intersection-free}, else we say that it has \emph{self-intersections}.
\end{definition}

For computation of self-intersections of a complex, one can consider all possible face-pairs and check if they intersect. One way to determine whether two faces have an intersection is to check if any of the edges of one of the faces intersects with the other face, and vice versa. By examining the computed intersection points, as seen in Section \ref{sec:detection}, we determine if an intersection is present or not.

\begin{definition}\label{def:retriangulation}
    Let $X,X'$ be two embedded complexes. We say that $X,X'$ are \emph{geometrically equivalent} if they give rise to the same chambers and in this case we call $X'$ a \emph{retriangulation} of $X$.
    We say that $X'$ \emph{rectifies} $X$ if both embedded complexes are geometrically equivalent and $X'$ has no self-intersections.
\end{definition}

In Section \ref{sec:outer_hull}, we show that we can compute the outer-hull of a complex $X$ without self-intersections. Then together with the embedding, the normal vectors are sufficient to create a 3D model of the simplicial complex, as seen in \cite{szilvsi2003analysis}. But if an embedding has a self-intersection, the following needs to be considered.
\begin{remark}
        \label{rem:outer_hull_intersections}
    The outer hull of an embedded complex $X$ is well-defined, even in the case that the embedding has self-intersections. However, then our description of $X$ via coordinates of vertices is not sufficient to compute the outer hull, as can be seen in Figure \ref{self-intersection}.This is because faces can lie in multiple chambers simultaneously, making it impossible to use normals to infer information about the outer hull. In this example, computing the outer hull using normal vectors, as described in Remark \ref{rem:oriented_normals}, is not possible. 
\end{remark}

\begin{figure}[H]
\begin{subfigure}{0.48\textwidth}
\includegraphics[height=5cm]{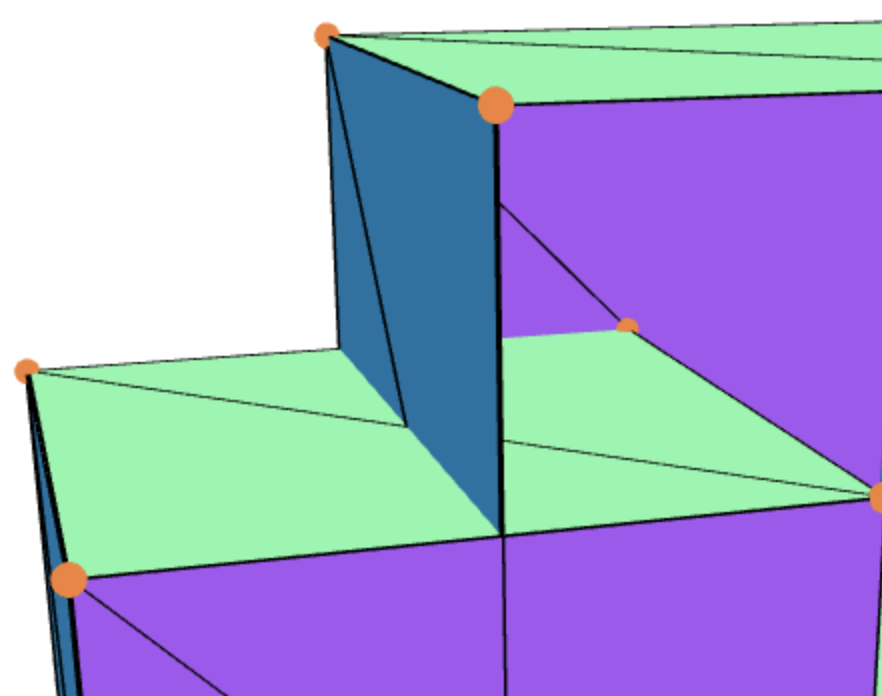}
\caption{Cut open view of two intersecting cubes.}
\end{subfigure}
\begin{subfigure}{0.48\textwidth}
\centering
\includegraphics[height=5cm]{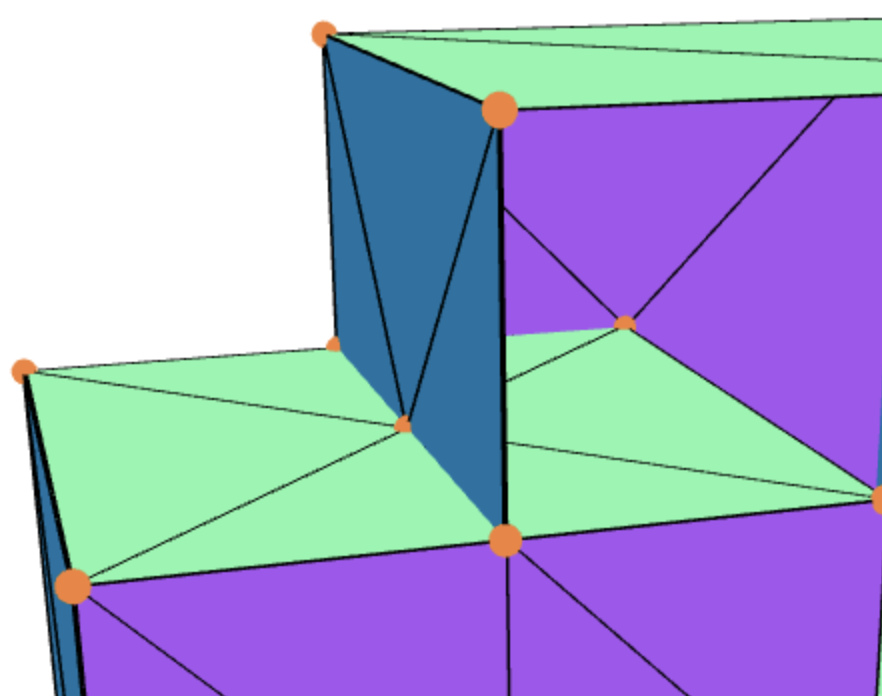}
\caption{Cut open view of two retriangulated intersecting cubes.}
\end{subfigure}
\caption{Intersection of triangles in two intersecting cubes. In (b), each face can be associated to exactly two chambers.}
\label{self-intersection}
\end{figure}

\section{Detecting Self-Intersections}
\label{sec:detection}

In order to retriangulate a complex with self-intersections, we first need to detect its intersection points (Definition \ref{def:intersection_points}). One objective of our framework lies in using a robust way for detecting self-intersections. For this, we employ a well-known method to check if two triangles intersect. It is sufficient to consider only two triangles at a time, and compare all such face-pairs one after another. Here, efficiency can also be gained using a-priori criteria to omit face-pairs that cannot intersect (for example based on vertex coordinates).\\ As our methods are modular, one could also use state of the art intersection detection algorithms such as PANG2 (see \cite{TriangleTriangleIntersection}) instead. Below, we present a numerical robust method for detecting self-intersecting triangles.

For this, we assume our initial data is a simplicial complex $X$ with embedding $\phi$ and consider two faces $f_\ell =\{v_1^\ell,v_2^\ell,v_3^\ell\},\ \ell=1,2$ in $X$. Here, we identify the vertices and edges of $X$ with their images under the embedding $\phi$.
We determine intersections points by examining the edges of $f_1$ via the normal equation of the plane $P(f_2)$ spanned by $v^2_1,v^2_2$ and $v^2_3$. Afterwards, we check if the found point lies inside the triangle $f_2$ or not. Finally, after computing all such intersection points of $f_1$ and $f_2$, we can use them to determine if the two triangles intersect.\\
We first fix our definitions before presenting the method in more detail.
There is an alternative characterisation to Definition \ref{def:contained} (containment in edges and faces) for points, which is more numerically robust. For this, we use the standard Euclidean dot product $\langle\cdot,\cdot\rangle$.
\begin{proposition}
    \label{prop:inside_triangle}
     Let $f$ be a triangle with vertices $v_1, v_2, v_3 \in \mathbb{R}^3$ and normal vector $n_f$, and consider the planes \begin{align*}
        P_1  := \mathrm{span}(n_f, v_1 - v_2), 
        P_2  := \mathrm{span}(n_f, v_2 - v_3), 
        P_3  := \mathrm{span}(n_f, v_3 - v_1).
    \end{align*}
    A point $p$ lies inside $f$ if and only if $p$ lies on the same plane as $f$ and satisfies the following condition:
    \begin{align}
        \label{eq:plane_inclusion}
        \langle p - v_i, \tilde{n}_i \rangle \geq 0, \quad \text{for } i \in \{1, 2, 3\},
    \end{align}
    where $\tilde{n}_i$ is a normal vector of the plane $P_i$, pointing towards the interior of $f$.
    
\end{proposition}
    This is visualised for the planar case in Figure \ref{fig:normal_equations}. In the sense of Inequality \eqref{eq:plane_inclusion}, $p_1$ has non-negative scalar product with all normal vectors, while $p_2$ has negative scalar product with $\Tilde{n}_2$.
\begin{figure}[H]
    \centering
    \begin{tikzpicture}[scale=5]
    \coordinate (v1) at (0,0,0);
    \coordinate (v2) at (1,0,0);
    \coordinate (v3) at (0.5,0.8,0);
    
    \draw[fill=blue!20] (v1) -- (v2) -- (v3) -- cycle;
    
    \coordinate (p) at (0.5,0.3,0);
    \coordinate (p2) at (0.5,0.5,0);
    \draw[fill=red] (p2) circle (0.5pt) node[below] {$p_1$};

    \coordinate (p3) at (0.8,0.6,0);
    \draw[fill=red] (p3) circle (0.5pt) node[below] {$p_2$};
    
    \coordinate (mid12) at ($(v1)!0.5!(v2)$);
    \coordinate (mid23) at ($(v2)!0.5!(v3)$);
    \coordinate (mid31) at ($(v3)!0.5!(v1)$);
    
    \draw[->, thick] (mid12) -- ($(mid12)!0.15cm!(p)$) node[below left] {$\Tilde{n}_1$};
    \draw[->, thick] (mid23) -- ($(mid23)!0.15cm!(p)$) node[below right] {$\Tilde{n}_2$};
    \draw[->, thick] (mid31) -- ($(mid31)!0.15cm!(p)$) node[below left] {$\Tilde{n}_3$};
    
    \foreach \v/\pos in {1/below left,2/below right,3/above}{
        \draw[fill=black] (v\v) circle (0.5pt) node[\pos=2pt] {$v_\v$};
    }
\end{tikzpicture}
    \caption{Triangle $f$ with vertices $v_1,v_2,v_3$, normals $\Tilde{n}_i$ pointing inwards and point $p_1$ inside $f$ and point $p_2$ outside $f$.}
    \label{fig:normal_equations}
\end{figure}
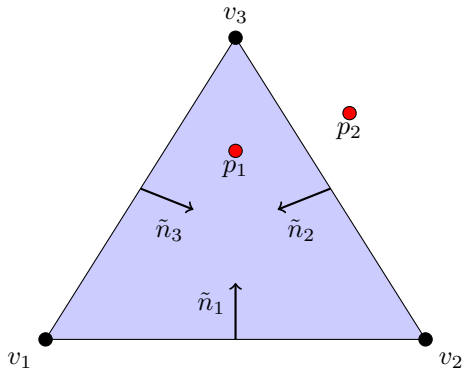
\begin{remark}
    Let $P$ be a plane with normal vector $n$ and let $x$ be a point in $P$. Proposition \ref{prop:inside_triangle} is based on the observation that if a point $p$ fulfils $\langle p-x, n\rangle \geq 0$, then $p$ lies above the plane (where the above direction is the one given by the orientation of the normal vector).
\end{remark}
Another important consideration are numerical errors.
\begin{remark}
    Consider the equation determining whether a point $p$ lies in a plane $P$, with normal $n$ and a fixed $x\in P$:
    \begin{align}
        \label{eq:plane_equation}
        p\in P \iff \scalp{p-x}{n} = 0.
    \end{align}
    Since numerical computations are not always precise, a point $p$ that does lie in $P$ will only yield $\scalp{p-x}{n} \approx 0$. Thus, for implementation purposes, one actually needs to check whether $\scalp{p-x}{n} \in [-\varepsilon,\varepsilon]$ for a small constant $\varepsilon > 0$. We apply this to all the conditions mentioned in this section, as can be seen in the code in our package.
\end{remark}

In the following, we give a detailed description of our method, in the setting as described above. First, we examine the edges of $f_1$. For an edge $\{v_i^1,v_j^1\}=e \subset f_1$, we have to distinguish the following two cases:
\begin{enumerate}[(i)]
    \item $e$ is not parallel to $f_2$.
    \item $e$ is parallel to $f_2$.
\end{enumerate}
Case (i) holds if and only if $\scalp{v^1_j-v^1_i}{n_{2}} \neq 0$, where $n_{2}$ is a unit normal vector of the face $f_2$.

 First, assume case (i). If $n_{2}$ is a normal vector of the plane $P(f_2)$, containment of a point $p$ in the plane $P(f_2)$ can be described as in Equation \eqref{eq:plane_equation}. Letting $\ell_{ij}$ denote the line connecting $v_i^1$ and $v_j^1$, parameterised by $\ell_{ij}(\alpha) := \alpha v_j^1 + (1-\alpha)v_i^1$, we thus get
\begin{align}
\label{eq:intersection}
\ell_{ij}(\alpha)\in P(f_2) \iff \alpha  = \dfrac{\scalp{v^2_1-v^1_i}{n_{2}}}{\scalp{v^1_j-v^1_i}{n_{2}}}.
\end{align}
Equivalence \eqref{eq:intersection} yields a criterion for deciding whether a point $\ell_{ij}(\alpha)$ on the line $\ell_{ij}$ lies in the plane $P(f_2)$. As we have assumed case (i), there exists $\alpha$ such that $\ell_{ij}(\alpha)\in P(f_2)$.
This point $\ell_{ij}(\alpha)$ is part of $f_1$ if $\alpha \in [0,1]$.
In case this holds, we also need to check if it is contained in the interior of the triangle $f_2$. For this, we use the geometric condition described in Proposition \ref{prop:inside_triangle}.

For case (ii), when the edge is contained in the plane $P(f_2)$, the problem reduces to a two-dimensional one and the following subcases need to be differentiated:
\begin{enumerate}[(i)]
    \item The edge is contained in a different but parallel plane as the triangle
    \item The edge is contained in the same plane as the triangle and one of the following holds
    \begin{enumerate}
        \item does not intersect the triangle,
        \item intersects the triangle,
        \item lies inside the triangle.
    \end{enumerate} 
\end{enumerate}
For a theoretical treatment of these subcases, we refer to \cite{vince2005geometry}, and for an implementation to our provided package \cite{Amend_Goertzen_2024}.

By using the methods described above, we compute all intersection points $I_2$ of edges of $f_1$ with $f_2$. We also need the reverse: the intersection points $I_1$ of edges of $f_2$ with $f_1$, as seen in Figure \ref{need-two-points}. \\
We also have to consider duplicates. Thus, we examine $I := I_1 \cup I_2$, as comparison of $I$ with the vertices of $f_1$ and $f_2$ and its cardinality $c:= |I|$ determines if the faces intersect.

\begin{figure}[H]
\centering
\begin{subfigure}{0.32\textwidth}
\centering
\includegraphics[height=3cm]{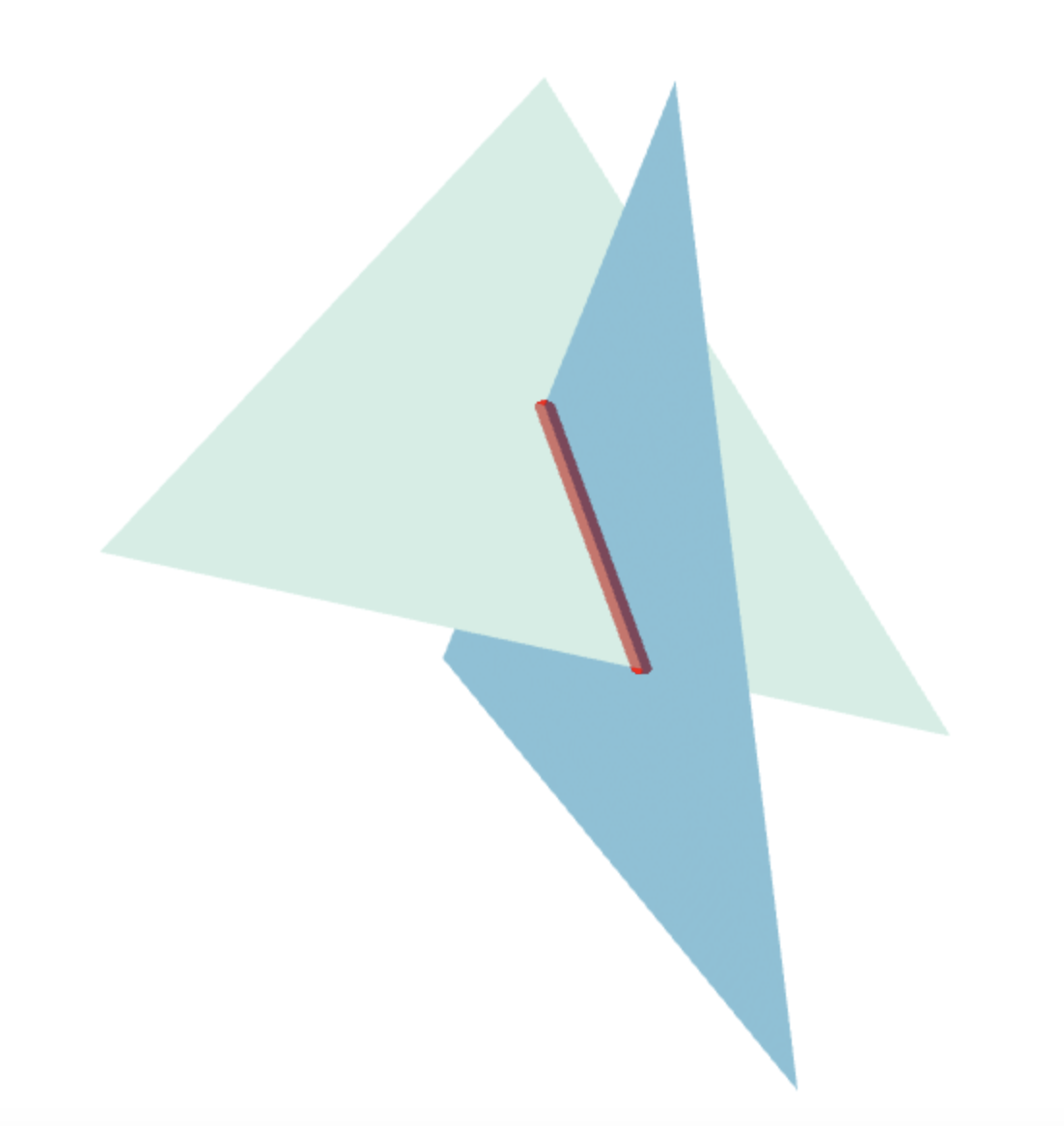}
\caption{Two intersection points}
\label{need-two-points}
\end{subfigure}
\begin{subfigure}{0.32\textwidth}
\centering
\includegraphics[height=3cm]{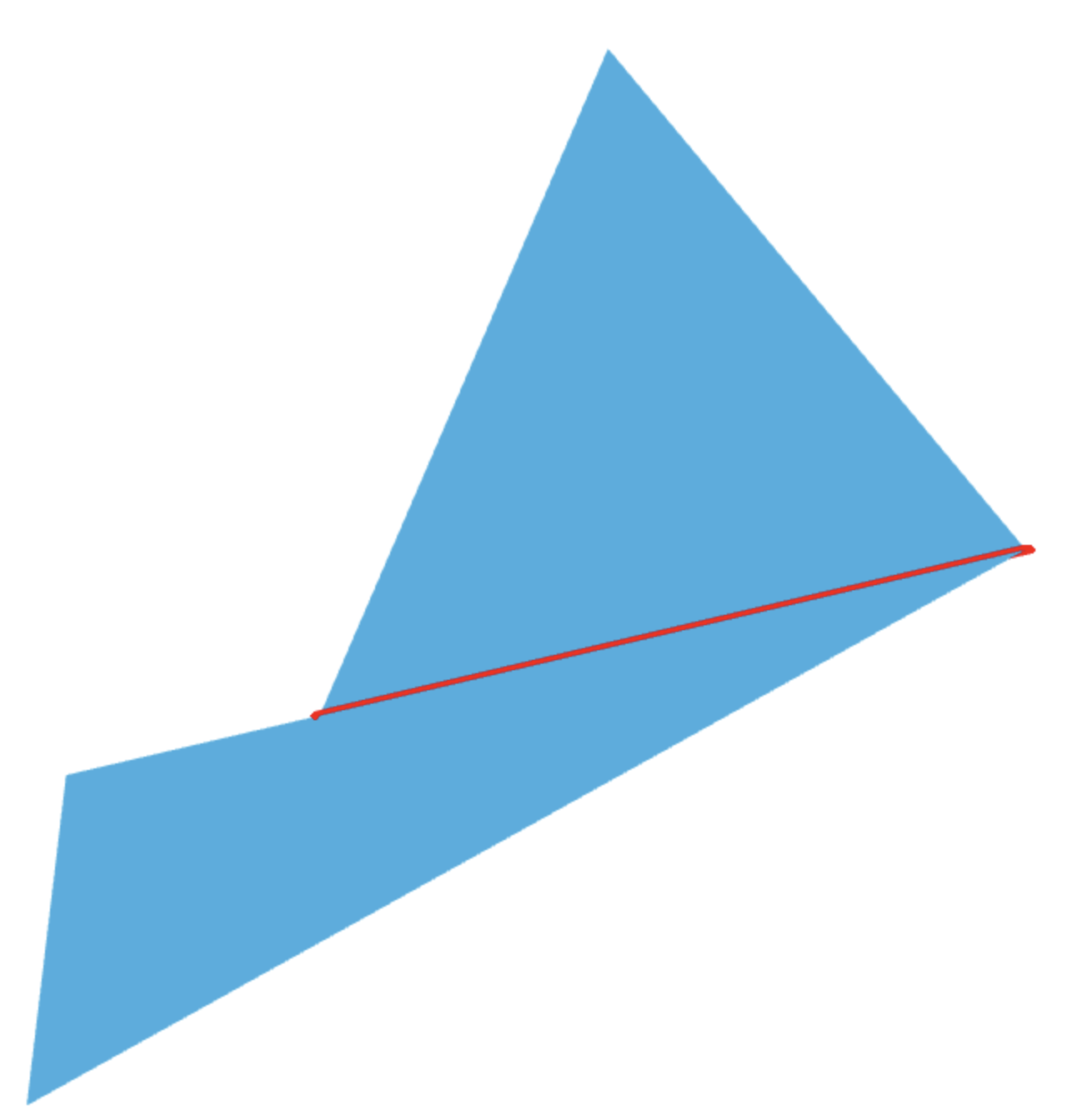}
\caption{Two intersection points}
\label{need-two-points2}
\end{subfigure}
\begin{subfigure}{0.32\textwidth}
\centering
\includegraphics[height=3cm]{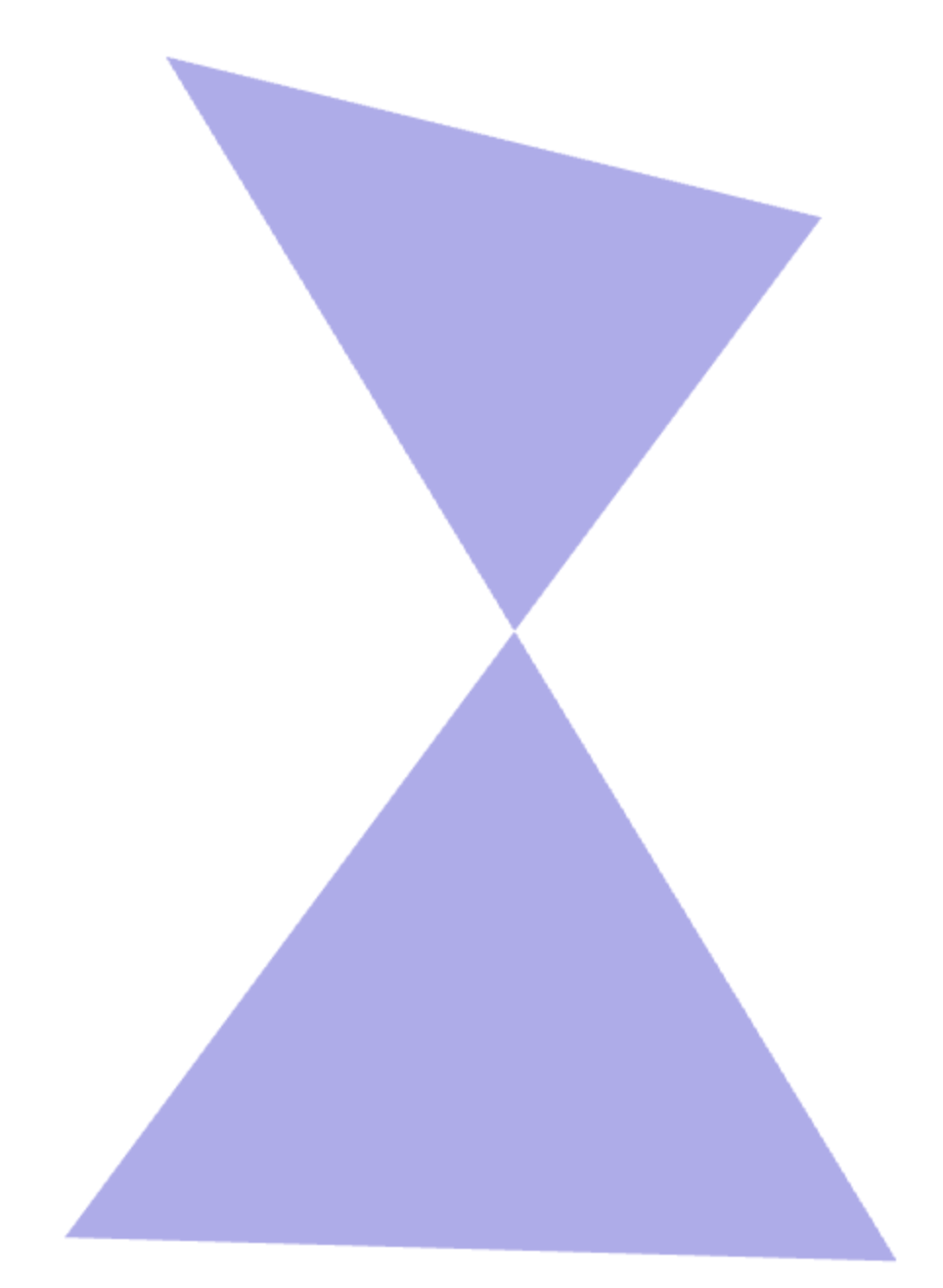}
\caption{Only one intersection point}
\end{subfigure}
\caption{Intersection of triangles (a, b) vs no intersection (c).}
\label{intersec-vs-none}
\end{figure}

An intersection is present if for instance $c = 2$ and $(I \setminus (f_1 \cup f_2)) \neq \emptyset$, as can be seen in Figures \ref{need-two-points} and \ref{need-two-points2}. The comparison with the original vertices of $f_1$ and $f_2$ is needed since adjacent triangles, compared along their joint edge, would also result in $c=2$. Thus, the intersection is parameterised by $v_1,v_2$ with $I_1 \cup I_2 = \{v_1,v_2\}$.\\
If $c=1$, and the intersection point is equal to a vertex of only one of the faces (or equal to none), there is also an intersection. We then parameterise it by $v$, where $I=\{v\}$.

After identifying all intersection points as listed above, we show in the following section how to retriangulate a face $f$ based on its intersection points. 

\newpage
\section{Retriangulation of Self-Intersecting Complexes}
\label{sec:rectif}

In this section, we introduce an algorithm that yields a retriangulation $X'$ of an embedded complex $X$ such that both complexes are geometrically equivalent and $X'$ has no self-intersections. The key idea is to give the retriangulation of each face a disc structure, with vertices and edges determined by the intersection points with other faces.

In Figure \ref{fig:ico_3_1_face_1}, we showcase how our method retriangulates a face $f\in X$ with self-intersections.

\begin{figure}[H]
\centering
\begin{subfigure}{0.32\textwidth}
\centering
\includegraphics[height=3.5cm]{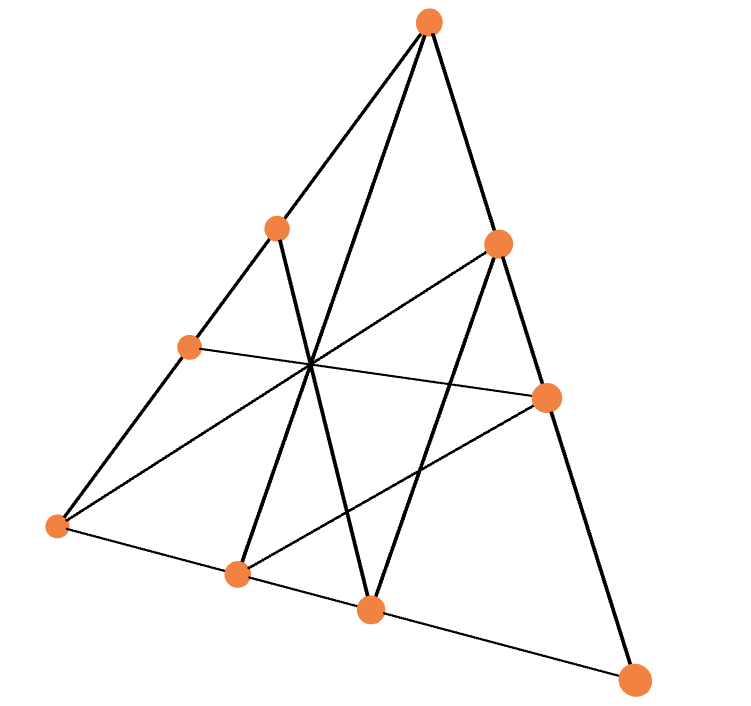}
\caption{}
\end{subfigure}
\begin{subfigure}{0.32\textwidth}
\centering
\includegraphics[height=3.5cm]{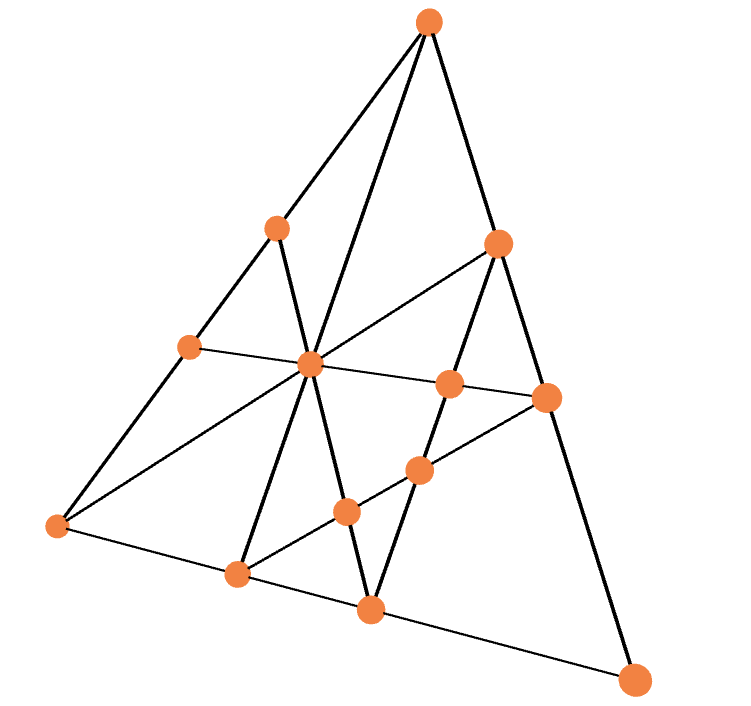}
\caption{}
\end{subfigure}
\begin{subfigure}{0.32\textwidth}
\centering
\includegraphics[height=3.5cm]{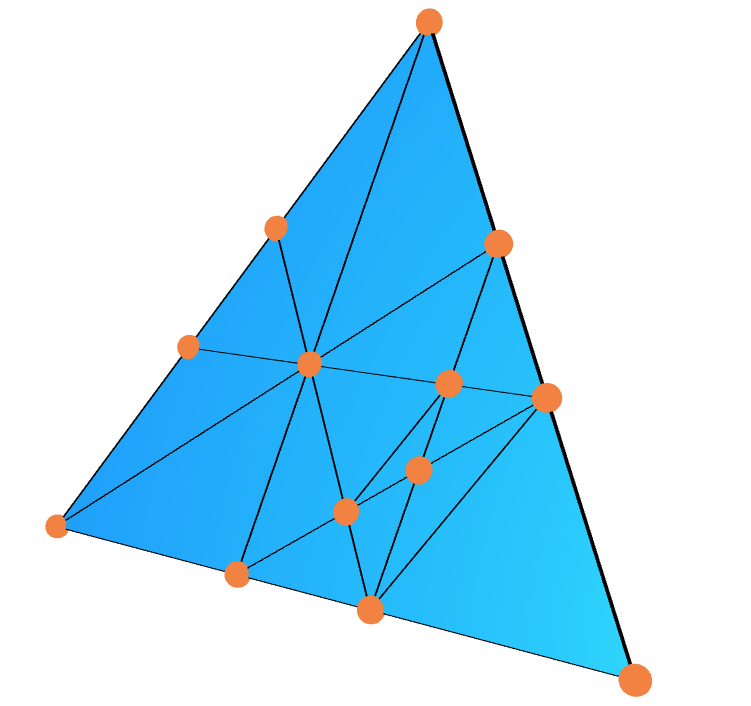}
\caption{}
\end{subfigure}
\caption{The steps of triangulating a face: (a) all intersections with other faces given as line segments, (b) finding intersection points within the line segments of the underlying face, (c) a simplicial disc giving a retriangulation of the original face.}
\label{fig:ico_3_1_face_1}
\end{figure}

\begin{remark}
    It suffices to retriangulate one face after another to obtain a retriangulation, see Definition \ref{def:retriangulation}. Hence, we restrict to the case of a single face $f\in X$ with intersections given by the set $I(X,\phi,f)$, where $(X,\phi)$ denotes the underlying embedded complex. 
\end{remark}

We proceed as follows: first, we check whether two intersections within $f$ with two other triangles $f',f''$ intersect and introduce new vertices if they do. Then, we retriangulate $f$ to obtain a simplicial disc $\Tilde{f}$ that is geometrically equivalent to $f$, i.e.\ spans the same area.

\begin{definition}
    A \emph{simplicial disc} $D$ is a simplicial surface with the following properties:
    \begin{enumerate}
        \item $D$ is connected with a single boundary component;
        \item $D$ is orientable.
    \end{enumerate}
\end{definition}

\begin{lemma}\label{lemma:disc_representation}
    A simplicial disc is uniquely described by the vertices of edges.
\end{lemma}

\begin{proof}
    This follows from the fact that for three edges $e_1,e_2,e_3$ that can be arranged in a closed cycle such that $(e_1,e_2)$, $(e_2,e_3)$ and $(e_1,e_3)$ share exactly one vertex, we obtain a unique face $f$ with incident edges given by $e_1,e_2,e_3$.
\end{proof}

Another characteristic of simplicial discs can be given as follows.

\begin{remark}
    For any simplicial disc $D$, there exists a simplicial sphere $S$ such that $D$ can be obtained by cutting the sphere $S$ into two parts along a simple closed vertex-edge cycle.
\end{remark}

We can compute the Euler Characteristic of a simplicial disc as follows.

\begin{lemma}
    The Euler characteristic of a disc $D$ is given by
    \begin{align} \label{euler_disc}
        V+(E+E')-\frac{E+(E+E')}{3}=1,
    \end{align}
where $V$ is the number of vertices, $E$ the number of inner edges and $E'$ the number of boundary edges of $D$.
\end{lemma}

\begin{proof}
    Gluing together two identical discs along their boundary leads to a sphere with $2V+V'$ vertices and $2E+E'$ edges, where $V'$ is the number of boundary vertices of the disc $D$. The number of faces in a closed simplicial sphere equals $2E/3$, and its Euler characteristic is equal to $2$. Using that the Euler characteristic of a simplicial sphere equals $2$, we obtain the following
    $$2V+V'-(2E+E')+\frac{2(2E+E'))}{3}=2.$$ Since the number of boundary vertices $V'$ of a disc equals the number of boundary edges $E'$, we arrive at the statement above.
\end{proof}

\begin{remark}
    Lemma \ref{lemma:disc_representation} yields that we can represent an embedded disc only by the vertices of edges and the coordinates of the given vertices. Thus, we introduce a data structure $\ell$ that contains all self-intersections of $f$ with other faces as a list with two entries:
    \begin{enumerate}
        \item The first entry is a list containing all coordinates of the embedded vertices $V$ of $f$, and all coordinates of the intersections with other faces.
        \item The second entry records the edges in a list that connect two vertices of the first list.
    \end{enumerate}
\end{remark}

Assuming that we computed all intersections of a given face $f$ with all other faces in the underlying complex, the starting step is to search for edge intersections and containment of vertices in edges inside $\ell$. For ease of notation, we define for a set $Y$ the set $Y^{(k)} :=\{g:\{1,\dots,k\}\rightarrow Y \text{ injective}\}$ as the set of $k$-tuples with pairwise distinct entries, for $k\in \mathbb{N}$.

In the following algorithm, intersecting lines within the given face $f$ are subdivided if they intersect each other. For this, we assume that we can remove duplicates of edges and vertices within $f$ using a helper function \textsc{CleanData}. The output of the algorithm contains the original face $f$ and subdivisions of all its intersections with other faces.

\begin{algorithm}[H]
\caption{RectifyDiscIntersections}\label{alg:fix_planar_intersections}
\Input{Triangle together with intersection points given by $\ell[1] \subset \mathbb{R}^3,\ell[2]\subset V^{(2)}$} 
\Output{Modified $\ell$ without self intersecting line segments} 
\Begin{
$\ell \gets \textsc{CleanData}(l)$\;
$edges \gets \ell[2]$\;
$vertices \gets \ell[1]$\;
\For{$v\in vertices$}{
    \For{$e \in edges$ }{
        \If{$v$ is contained in $e$}{
            $edges \gets edges\setminus e \cup \{\{v,e[1] \},\{v,e[2] \} \}$\;  \tcc{edge $e$ is split via vertex $v$}
        }
    }
}
$\ell \gets \textsc{CleanData}(\ell)$\;
$edgepairs \gets edges^{(2)}$\;
\For{$edgepair \in edgepairs$}{
    \If{$edgepair[1]$ intersects $edgepair[2]$ in $v \notin \ell[1]$}{
        $edges \gets edges\setminus \{edgepair[1],edgepair[2]\} $\;
        $edges \gets edges \cup \{\{v,edgepair[2][1] \},\{v,edgepair[2][2] \}$\;
        $edges \gets edges \cup \{\{v,edgepair[1][1] \},\{v,edgepair[1][2] \} \}$ \;
        \tcc{edges $edgepair[1],edgepair[2]$ intersect in a vertex $v$}
    }
}
$\ell \gets \textsc{CleanData}(l)$\;
\Return $\ell$\;
}
\end{algorithm}

The output of \refalg{alg:fix_planar_intersections} yields a data structure $\ell$ that describes a non-intersecting polygon together with inner parts given as line segments. Next, we want to triangulate this output. For this, it suffices to add non-self intersecting edges until condition \eqref{euler_disc} on the Euler characteristic of a simplicial disc is fulfilled. We add minimal non-intersecting edges until we obtain a triangulation. It follows that the resulting triangulation has a minimal number of vertices, as only necessary vertices are added. This algorithm can also be modified to consider more general polygons instead of triangles.

In order to apply condition \eqref{euler_disc} for giving a termination criterion, we first compute the number of boundary vertices $V'$ of $\ell$, which equals the number of boundary edges $E'$, and then use that $E+E'=\lvert \ell[2]\rvert,V=\lvert \ell[1] \rvert$. Since we only need to add inner edges, we can compute the number of total internal edges as
$$E=3V-2E'-3.$$
Therefore we only need to add $3\cdot \lvert \ell[1]\rvert-2E'-3-\lvert \ell[2]\rvert$ edges. This yields an improvement to the greedy algorithm of adding shortest edges, first proposed in \cite{duppe1970automatischegreedy}.

\begin{algorithm}[H]
\caption{DiscTriangulation}\label{alg:triangulation}

\Input{$\ell$ without self intersecting line segments} 
\Output{A simplicial disc $\ell$} 
\Begin{
    $possibleedges \gets V^{(2)}$\;
    $possibleedges \gets SortByLength(possibleedges)$\;
    \While{$V-(E+E')+\frac{E+(E+E')}{3}\neq1$}{
        $e_1 \gets \text{possibleedges}[1]$\;
        $possibleedges \gets \text{possibleedges} \setminus e_1$\;
        $foundnewedge \gets \text{true}$\;
        \For{$e_2 \in \ell[2]$}{
            \If{$e_1$ intersects $e_2$}{
                $foundnewedge \gets \text{false}$\;
                \tcc{edges $e_1$ and $e_2$ intersect}
                break\;
            }
        }
        \If{foundnewedge}{
                $\ell \gets \ell \cup \{e_1\}$
            }
    }
    \Return $\ell$\;
}
\end{algorithm}

\refalg{alg:triangulation} yields an approximation of a \emph{minimal weight triangulation}, i.e.\ summing up all edge lengths of the resulting disc is minimal. In general, solving this problem proves to be NP-hard \cite{MinimumWeightTriangulationNPHard}. Since each polygon without self-intersection in two dimensions can be triangulated without introducing new vertices \cite{ComputationalGeometry},  \refalg{alg:triangulation} terminates for arbitrary inputs.  

\section{Symmetric Optimisation for Retriangulation} \label{sec:symmetry}

In this section, we consider all orthogonal transformations (reflections and rotations) that leave a given embedded complex $X$ invariant in order to optimise computing its retriangulation. For this, we need to define the orthogonal group $\text{O}(3)$ that acts on $\mathbb{R}^3$.

\begin{definition}
    The orthogonal group $\text{O}(3)$ defined by rotation and reflection matrices acting on the real-space $\mathbb{R}^3$ consists of $3\times 3$-matrices $A$ with $A^{tr}\cdot A=\mathbb{I}$, i.e.\ matrices $A$ with inverse given by their transposed matrix.
\end{definition}

The standard scalar product of two vectors $x,y\in \mathbb{R}^3$ is invariant under multiplication by an orthogonal matrix $A$, since we have
    \[
    \langle Ax,Ay\rangle=\left(Ax\right)^{tr}\left(Ay\right)=x^{tr}A^{tr}Ay=x^{tr}y=\langle x,y\rangle.
    \]

Note that an embedded complex $X$ is determined by the coordinate vectors of its vertices $V$, together with the combinatorial information given by subsets of $P_3(V)$. From now on, we identify $V$ with a subset of $\mathbb{R}^3$. The action of an element  $\pi \in \text{O}(3)$ viewed as a map $\mathbb{R}^3\to\mathbb{R}^3$ on an embedded complex $X$ is determined by applying $\pi$ to the 3D coordinates of each vertex.  With this, we can define the symmetry group of $X$.

\begin{definition}
    The symmetry group of $X$ is defined as the group $\mathrm{Sym}(X)\leq \mathrm{O}(3)$ of all orthogonal transformations leaving $X$ invariant, i.e.\ $$\mathrm{Sym}(X) \coloneqq \{\pi \in \mathrm{O}(3) \mid \pi(X)=X \}.$$   
\end{definition}

The symmetry group $\mathrm{Sym}(X)$ of a simplicial complex $X$ acts on the embedded vertices and this action can be extended to edges and faces in a canonical way, i.e.\ if $e$ is an edge with embedded vertices $\{v_1,v_2 \}$ and $f$ is a face with embedded vertices $\{v_1,v_2,v_3\}$, we can define 
$$\pi (e)=\{\pi(v_1),\pi(v_2) \}\text{ and }\pi (f)=\{\pi(v_1),\pi(v_2),\pi(v_3) \},$$
for $\pi \in \mathrm{Sym}(X)$.
In this way, the group $\mathrm{Sym}(X)$ induces subgroups of the automorphism group $\text{Aut}(X)$ of the underlying simplicial complex $X$, i.e.\ permutation subgroups acting on the vertices, edges and faces. These permutation groups can be used when computing face and face-pair orbits. Switching to a discrete permutation group allows both faster and more precise computations without the use of numerical methods.

Before rectifying the self-intersections, as in Section \ref{sec:rectif}, we can determine a \emph{transversal} (also called \emph{face transversal}) of the orbits of $\mathrm{Sym}(X)$ on the set of faces, which is a minimal set of faces $\{f_1,\dots,f_n\}$ with $$\bigcup_{i=1}^n \mathrm{Sym}(X)(f_i)=X_2. $$Thus each face of the complex can be obtained by applying certain symmetries to a face in the much smaller set $\{f_1,\dots,f_n\}$. For instance, if $\mathrm{Sym}(X)$ acts transitively on $X_2$, we have $n=1$.

When computing self-intersection, we have to consider pairs of faces $(f_1,f_2)$. We can extend the action of $\mathrm{Sym}(X)$ on the faces $X_2$ to the pairs of distinct faces, denoted by $X_2^{(2)}$. Hence, it suffices to consider the orbits defined by this group action instead of considering all face pairs, $$\mathrm{Sym}(X)\times X_2^2 \to X_2^2,\ \big(\pi,(f_1,f_2)\big)\mapsto (\pi(f_1),\pi(f_2)).$$ Here, the faces $f_1$ and\ $f_2$ are mapped to $\pi(f_1)=f_1'$ and  $\pi(f_2)=f_2'$, respectively, by applying an element $\pi$ in the symmetry group $\mathrm{Sym}(X)$. Now we can observe that the faces $f_1,f_2$ intersect if and only if the faces $f_1',f_2'$ intersect. Let $I_{f_1,f_2}=f_1\cap f_2$ denote the intersection of two faces. With the notation above, it follows that $$\pi(I_{f_1,f_2})= I_{f_1',f_2'}.$$ This means that it suffices to consider only one element of the orbit $$\mathrm{Sym}(X)((f_1,f_2))\coloneqq\{(\pi(f_1),\pi(f_2)) \mid \pi\in \mathrm{Sym}(X) \}$$ to determine all intersections of face pairs in this set.

Combining both observations of the group actions of $\text{Sym}(X)$ on the faces $X_2$ and face pairs $X_2^{(2)}$, it suffices to consider only face pairs of the form $(f_i,f)$ with $i=1,\dots,n$ and $f_i$ being part of the chosen face transversal. In order to obtain all such relevant face pairs, we can compute the orbits of the stabiliser $\text{Stab}_{f_i}(X_2^{(2)})$ of the face $f_i$ on the set of faces $X_2$.

In summary, the symmetry group of an embedded complex $X$ can be used in two ways to simplify computations for self-intersections:
\begin{enumerate}
    \item Finding self-intersections: given the self-intersections of a pair $\{f_1,f_2 \}$ of faces, the self-intersections of all pairs of faces in the orbit of $\{f_1,f_2 \}$ under the symmetry group are also known.
    \item Retriangulation: it is only necessary to fix self-intersections and retriangulate one face in each orbit of the symmetry group acting on the faces.
\end{enumerate}

These steps are summarised in \refalg{alg:symmetry}.

\begin{algorithm}[H]
\caption{SymmetricRetriangulation}\label{alg:symmetry}

\Input{Embedded complex $X$ and $\text{Sym}(X)$} 
\Output{Embedded complex $X'$ without self-intersection} 
\Begin{
FaceRepresentatives$=$ representatives of orbits of $\mathrm{Sym}(X)$ on the faces $X_2$ of $X$\;
FacePairRepresentatives$=[\ ]$\;
\For{$f\in$ FaceRepresentatives}{
    \tcc{Add all face pair representatives containing $f$}FacePairRepresentatives$[f]=\text{representatives of orbits of }\text{Stab}_{f}$ acting on face pairs $(f,f')$\;
}
FaceIntersectionsRepresentatives$=[\ ]$\;
\For{$f\in$ FaceRepresentatives}{
    \For{$(f,f')\in $ FacePairRepresentatives$[f]$}{
        \If{HasIntersection($(f,f')$)}{
            \tcc{Find all representative face pair intersections with $f$} 
            Add Intersection($(f,f')$ to the list FaceIntersectionsRepresentatives$[f]$\;
        }
    }
}
\tcc{Transfer all representative face pair intersections containing $f$}
FaceIntersections$=[\ ]$\;
\For{$f\in$ FaceRepresentatives}{
    \For{$\pi \in \text{Stab}_{f}$}{
        FaceIntersections$[f]=$ Concatenate(FaceIntersections$[f]$,$\pi$(FaceIntersectionsRepresentatives$[f]$))\;
    }
}
\tcc{Fix intersections and retriangulate face representatives}
FaceRetriangulations$=[]$\;
\For{$f\in$ FaceRepresentatives}{
    FaceRetriangulations$[f]=$Retriangulate($f$)\;
    \tcc{transfer retriangulation to remaining faces}
    \For{$\pi \in \text{Stab}_{f}$}{
          FaceRetriangulations$[\pi (f)]=\pi ($FaceRetriangulations$[f])$\;        
    }
}
\Return FaceRetriangulations

}
\end{algorithm}

The resulting speedup compared to methods that do not exploit the symmetry of the given complex is proportional in both steps to the number of orbits. The number of orbits for a group acting on a given set can be computed using a well-known result attributed to Cauchy, Frobenius and Burnside:

\begin{lemma}[Cauchy-Frobenius-Burnside Lemma]
    The number of orbits $|\Omega/G|$ of a finite group $G$ acting on a finite set $\Omega$ equals the average amount of fixed points, i.e.\ $$|\Omega/G|=\frac{1}{|G|} \sum_{g\in G} \mathrm{Fix}_g(\Omega),$$where $\mathrm{Fix}_g(\Omega)=\lvert \{\omega \in \Omega \mid g(\omega )=\omega  \}\rvert$ is the number of fixed points of a group element $g$  acting on $\Omega$.
\end{lemma}

For example, if only the identity element of $\mathrm{Sym}(X)$ fixes elements of the set $X$, we have exactly $|X|/|\mathrm{Sym}(X)|$-orbits. In terms of \refalg{alg:symmetry}, this would result in a speedup proportional to the order of $\mathrm{Sym}(X)$, i.e.\ the number of symmetries leaving $X$ invariant.

Another advantage of our approach is that we can preserve the symmetry of the resulting self-intersection complex. 
\begin{remark}
    We can obtain an embedded complex $X'$ without self-intersections and $$ \text{Sym}(X)=\text{Sym}(X'),$$
    if we account for the preservation of local symmetries of faces, i.e.\ if $f$ is a face in $X$ and $D_f$ is the resulting disc obtain by retriangulation $f$ together with its intersection with other faces we need to enforce for $\pi \in \text{Sym}(X)$ with $\pi(f)=f$ that $\pi(D_f)=D_f$.
    This can be combined with a more general approach describing local symmetry-preserving operations, as given in \cite{LocalPreservingSymmetryApproach,LocalPreservingSymmetryMethod}.
\end{remark}

\section{Computation of Outer Hull and Chambers} \label{sec:outer_hull}

For several purposes, such as 3D printing or surface modelling, it is necessary to compute the outer-hull of a self-intersecting complex $X$ to disregard inner parts. For this, we consider an algorithm introduced in  \cite{Attene2018}. This algorithm relies on an initialisation, for which we provide a proof below, with an outer triangle $t$ (so one that is part of the outer hull of $X$) with normal $n$ pointing outwards. Here, we also write $(t,n)$ for a pair of a triangle and its normal.  Moreover, we show that this algorithm can be also applied to obtain any chambers within the complex.\\

\begin{lemma}
\label{lemma:initialize_outer}
Let $X$ be a closed simplicial complex with vertices given by $V$ and embedding $\phi:V\to\mathbb{R}^{3}$ such that $\left(X,\phi\right)$ has no self-intersections. We identify vertices, edges and faces with their embedding under $\phi$ and for $p\in \mathbb{R}^3$, we write $p=(p_x,p_y,p_z)$. 
Now, let $\Tilde{v}$ be the vertex with maximal first-coordinate, i.e.\
\[
\Tilde{v}_{x}=\max\left\{v_{x}\mid v\in V\right\}.
\]
Let $\Tilde{e}=\{\Tilde{v},v' \}$ be the edge in $X$ incident to $\Tilde{v}$ such that
\[
\frac{\lvert(\Tilde{v} - v')_{x}\rvert}{\lVert \Tilde{v} - v'\rVert} = \min\left\{\frac{\lvert(\Tilde{v} - v)_{x}\rvert}{\lVert \Tilde{v} - v\rVert} \;\middle|\; \Tilde{v}\subset e=\{\Tilde{v},v'\}\right\}
\] (compare Figure \ref{lemma_init}).
Let $\Tilde{t}$ be the face incident to $\Tilde{e}$ with unit normal $\Tilde{n}=(\Tilde{n}_1,\Tilde{n}_2,\Tilde{n}_3)$
such that
\[
\lvert \Tilde{n}_{x}\rvert=\max\left\{\lvert n_{x}\rvert\mid n\text{ is a unit normal of a face }f\text{ with }\Tilde{v} \subset \Tilde{e} \subset f\right\}.
\]
If $\Tilde{n}_{x}<0$, negate $\Tilde{n}$. Then $\Tilde{t}$ is an outer triangle  with normal $\Tilde{n}$ pointing outwards.
\end{lemma}

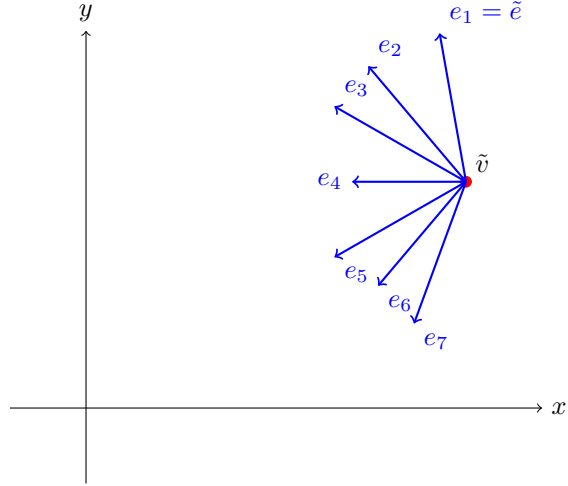
\begin{figure}[H]
\centering
\begin{tikzpicture}
    \draw[->] (-1,0) -- (6,0) node[right] {$x$};
    \draw[->] (0,-1) -- (0,5) node[above] {$y$};

    \filldraw[red] (5,3) circle (2pt) node[above right, black] {$\Tilde{v}$};

    \draw[blue, thick, ->] (5,3) -- ++(100:2) node[above right] {$e_1=\Tilde{e}$};
    \draw[blue, thick, ->] (5,3) -- ++(130:2) node[above right] {$e_2$};
    \draw[blue, thick, ->] (5,3) -- ++(150:2) node[above right] {$e_3$};
    \draw[blue, thick, ->] (5,3) -- ++(180:1.5) node[left] {$e_4$};
    \draw[blue, thick, ->] (5,3) -- ++(210:2) node[below right] {$e_5$};
    \draw[blue, thick, ->] (5,3) -- ++(230:1.8) node[below right] {$e_6$};
    \draw[blue, thick, ->] (5,3) -- ++(250:2) node[below right] {$e_7$};
\end{tikzpicture}
\caption{Projection of unit edge vectors onto the plane $z=0$.}
\label{lemma_init}
\end{figure}

\begin{proof}
     Since the complex is intersection free, either $\Tilde{n}$ or $-\Tilde{n}$ is a corrected orientated normal if we know that $\Tilde{t}$ lies on the outer hull. Additionally, we may assume that, up to rotation, the $x$-value of $\Tilde{v}$ is a strict maximum. \\
     So assume $(\Tilde{t},\Tilde{n})$ to be incorrectly oriented, thus either $\Tilde{t}$ is not part of the outer hull, or $\Tilde{n}$ is pointing inwards. Then from all points $p\in \Tilde{t}$, we hit another point $h_p\in X$ when travelling along the line $p+\alpha \Tilde{n}$, $\alpha \geq 0$. Since the complex is intersection-free, if $p$ is not a vertex, then $h_p \neq p$. \\
     Label the vertices of $\Tilde{t}$ such that the edge $\Tilde{e}$ connects $\Tilde{v}$ and $v'$, and call $\hat{e}$ the edge connecting the third vertex $v_2$ to $\Tilde{v}$. Then take a sequence of points $\ell_n \in \hat{e}$ that converges to $\Tilde{v}$. The corresponding sequence $h_{\ell_n}$ also converges to $\Tilde{v}$ by maximality of $\Tilde{v}$ and positive $x$-coordinate of $\Tilde{n}$. Since we only have finitely many faces in our model, there is a minimal face diameter, thus we can assume that without loss of generality, $h_{\ell_n}\in t'$ for some face $t'$. The fact that $X$ is intersection-free forces $\Tilde{v} \in t'$ and since we can hit $t'$ along $\hat{e}$ via $\Tilde{n}$, the other two vertices of $t'$ have to have $x$-value greater than $v_2$. Thus, $\Tilde{e}$ can be reached via two edge-face paths from $t'$, since $X$ is a closed surface. In at least one direction, there can never be vertices that do not have $x$-value greater than $v_1$ as else there would need to be a self-intersection present. Thus, there exists a vertex $v$ with $x$-value greater than $v_1$ that is incident to a face $t''$ that contains $\Tilde{e}$ (so that $v$ is not incident to $\Tilde{e}$). But this contradicts the choice of the normal $\Tilde{n}$, so $(\Tilde{t},\Tilde{n})$ is correctly oriented.
\end{proof}

The idea of computing the outer hull is based on depth-first search (DFS) on the faces and edges of the outer hull.  For this, we use the initialisation configuration given in Lemma \ref{lemma:initialize_outer} above for a given face and a normal vector. Next, we need a criterion to decide for each edge which face lies on the surface of the outer hull. In \cite{Attene2018}, Attene introduces an algorithm for the computation of the outer-hull of a complex $X$ without self-intersections. For each edge $e$ in the complex, a \emph{fan of faces} $\rm{Fan}(e)$, consisting of all faces incident to $e$ in the complex, can be defined.
\begin{definition}
    Let $X$ be a complex and $e$ an edge of $X$. We define the fan of $e$ as the set
    \begin{align*}
        \mathrm{Fan}(e) \coloneqq \{ f \mid e \subset f \text{ face in } X\}.
    \end{align*}
\end{definition}
An ordering of $\rm{Fan}(e)$ can be obtained via a face $f_1 \in \rm{Fan}(e)$ and its normal vector $n$, see Figure \ref{fig:FAN}. For this, we compute the \emph{upward continuation} of the face $f_1$, which is defined as the next face $f_2 \in \mathrm{Fan}(e)$ in direction of $n$. We can associate a normal to $f_2$ by rotating the normal vector $n$ along the edge $e$ and negating it. The concept of a fan is also used in Section \ref{sec:ramifications} when considering non-manifold parts after computing the outer hull.

Since we start with a correct initialisation step and the surface is connected, we can proceed inductively to cover the whole surface of the outer-hull. In \cite{Attene2018}, Attene also covers sheets, i.e.\ faces that are contained with both sides inside the outer-hull. Because the initial surface is closed, we do not need to consider this case here. 

\begin{figure}[H]
\centering
\tdplotsetmaincoords{70}{110}
\begin{tikzpicture}[tdplot_main_coords, scale=3]

    \coordinate (A) at (0,0,0);
    \coordinate (B) at (1,0,0);

    \coordinate (C) at (0.5, 0.866025, 0);     
    \coordinate (D) at (0.5, -0.432966, 0.750027);  
    \coordinate (E) at (0.5, -0.866025, 8.02404e-05);  
    \coordinate (F) at (0.5, 0.432897, -0.750067);  

    \foreach \first/\second in {C/D, D/E, E/F, F/G}{
        \draw[fill=blue, fill opacity=0.5] (A) -- (B) -- (\first) -- cycle;
    }

    \draw[fill=green, fill opacity=0.8] (A) -- (B) -- (C) -- cycle;
    \draw[thick, red] (A) -- (B); 

    \node at (C) [above right] {$f_1$}; 
    \node at (D) [above right] {$f_2$};
    \node at (E) [below left] {$f_3$};
    \node at (F) [left] {$f_4$};

    \node at ($(A)!0.5!(B)$) [below] {$e$};

    \draw[->, thick] (0.5,0.433013,0) -- +(0,0,0.5) node[above] {$\mathbf{n}$}; 

\end{tikzpicture}
\caption{For the edge $e$, $\rm{Fan}(e)$ is given by four faces $f_1,f_2,f_3,f_4$. Using the normal vector $n$ of face $f_1$, we can order the faces in a circle $(f_1,f_2,f_3,f_4)$. The upward continuation of the face $f_1$ with normal $n$ is then given by $f_2$.}
\label{fig:FAN}
\end{figure}
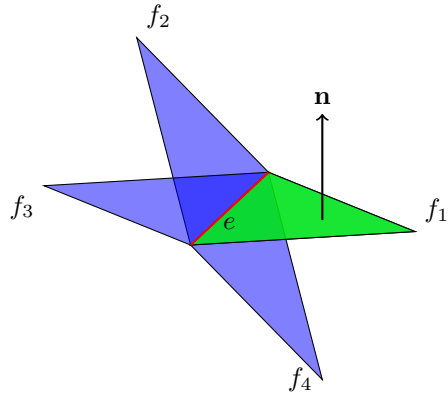
We extend the algorithm in \cite{Attene2018} to compute not only the outer hull, but all chambers of $X$ given an initialization face $t_0$ with normal $n_0$.

\begin{algorithm}[H]
\caption{ExtractChamber}\label{alg:chamber}

\Input{Intersection-free complex $X$, initial face $t_0$ with normal $n_0$} 
\Output{Chamber containing face $t_0$ with side given by $n_0$} 
\Begin{
$e_1,e_2,e_3=$ three edges of initialization face $t_0$\;
$B=[\langle t_0,e_1 \rangle,\langle t_0,e_2 \rangle,\langle t_0,e_3 \rangle$]\;
mark $t_0$ as outer\;
\While{$B\neq []$}{
     $\langle t_i,e_j \rangle=$first element of $B$\;
     remove $\langle t_i,e_j \rangle$ from $B$\;
     $t_{new}=$ upward continuation of $t_i$ at $e_j$\;
     \If{$t_{new}=$ is not tagged as outer}{
        mark $t_{new}$ as outer\;
        $e_1,e_2=$two edges of face $t_{new}$ not equal to $e_j$\;
        Add $\langle t_{new},e_1 \rangle,\langle t_{new},e_2 \rangle$ to the beginning of $B$\;
     }
}
}
\end{algorithm}

For any face $f$ of $X$ with given normal vector $n$, we obtain a closed connected subcomplex $Y$ of $X$ using \refalg{alg:chamber}, which is the boundary of a chamber of $X$ (see Definition \ref{def:outer_hull}). Here, \refalg{alg:chamber} associates a set of faces to a chamber. 

\begin{lemma}\label{chamber_partition}
    The set $X_2$ of faces of an embedded simplicial complex $X$ together with a two-sided orientation of each face can be partitioned using the notion of chambers above, i.e.\ each side of each face belongs to a unique chamber. The outer-hull can be obtained using the start configuration from Lemma \ref{lemma:initialize_outer}.
\end{lemma}

\begin{proof}
     \refalg{alg:chamber} stays inside a given chamber. As a chamber is connected and a depth-first search is conducted on the faces in a chamber, the output is independent of the initialisation face of a given chamber.
\end{proof}

Computing chambers and visualising them leads to a more profound understanding of the involved complexes. Volumes and other invariants can be computed in parallel for the whole  complex by using a decomposition into chambers.

\begin{example}
   Consider the self-intersecting complex $X_{23}$, shown in Figure \ref{fig:x23_cyclic}, with symmetry group isomorphic to $C_{23}$ (cyclic group with $23$ elements). It belongs to an embedded family of simplicial surfaces with equilateral triangles and cyclic symmetry, see \cite{akpanya2023surfaces} for exact coordinates. The complex $X_{23}$ has $6\cdot 23=138$ faces, $71$ vertices and $207$ edges and thus Euler characteristic $71-207+138=2$. The embedding of the vertices of $X_{23}$ is chosen in a way that its symmetry group $G=\text{Sym}(X_{23})\leq \text{O}(3)$ is given by elements of the form
    $$\begin{pmatrix}\cos(\varphi) & -\sin(\varphi) & 0\\
    \sin(\varphi) & \cos(\varphi) & 0\\
    0 & 0 & 1
    \end{pmatrix},$$
    where $\varphi \in \{ \frac{2\pi \cdot j}{23} \mid 0\leq j < 23\}$ and it is generated by the rotation matrix $$\begin{pmatrix}\cos( \frac{2\pi }{23}) & -\sin(\frac{2\pi }{23}) & 0\\
    \sin(\frac{2\pi }{23}) & \cos(\frac{2\pi }{23}) & 0\\
    0 & 0 & 1
    \end{pmatrix},$$ which can be seen by matrix-multiplication and the addition theorem of $\cos$ and $\sin$. By labelling the $138$ faces by $f_1$ to $f_{138}$, we obtain $\binom{138}{2}=9453$ face pairs. There are exactly $6=\frac{138}{23}$ orbits of faces under the action of $G$, with representatives that we call $f_1$ to $f_6$. We can find all face pair intersections by the action of $G$ on pairs of the form $(f_i',f),i=1,\dots,6$, where $f$ is a face of $X_{23}$. Hence, it suffices to test intersection on these orbit representatives, as every intersection can be obtained by applying a group element to one of these face pairs (see Section \ref{sec:symmetry}). We retriangulate $f_1,\dots,f_6$ using \refalg{alg:triangulation} and then apply $G$ to obtain a retriangulation of the entire surface $X_{23}$. Transferring, i.e.\ applying group elements of $G$, can be achieved in time linear in the number of faces. Thus, we obtain an approximate speedup in the order of $G$, so by a factor of $23=|G|$. In Figure \ref{c23_explode}, we see an exploded view of the internal chambers of $X_{23}$. 
\end{example}

\begin{figure}[H]
\begin{subfigure}{0.3\textwidth}
\includegraphics[height=5cm]{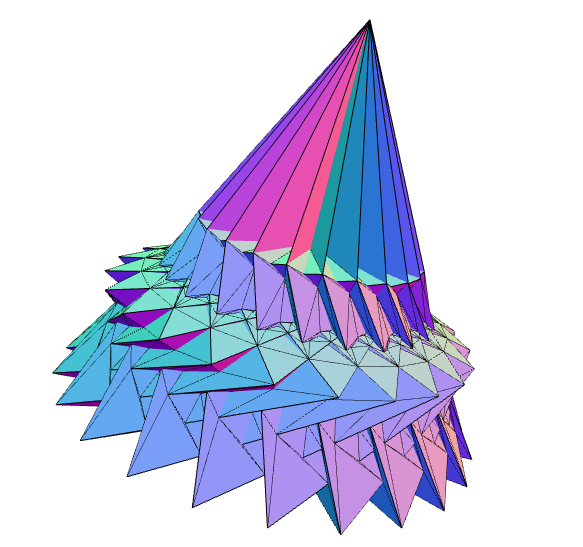}
\caption{$m=0$}
\label{fig:x23_cyclic}
\end{subfigure}
\begin{subfigure}{0.3\textwidth}
\includegraphics[height=5cm]{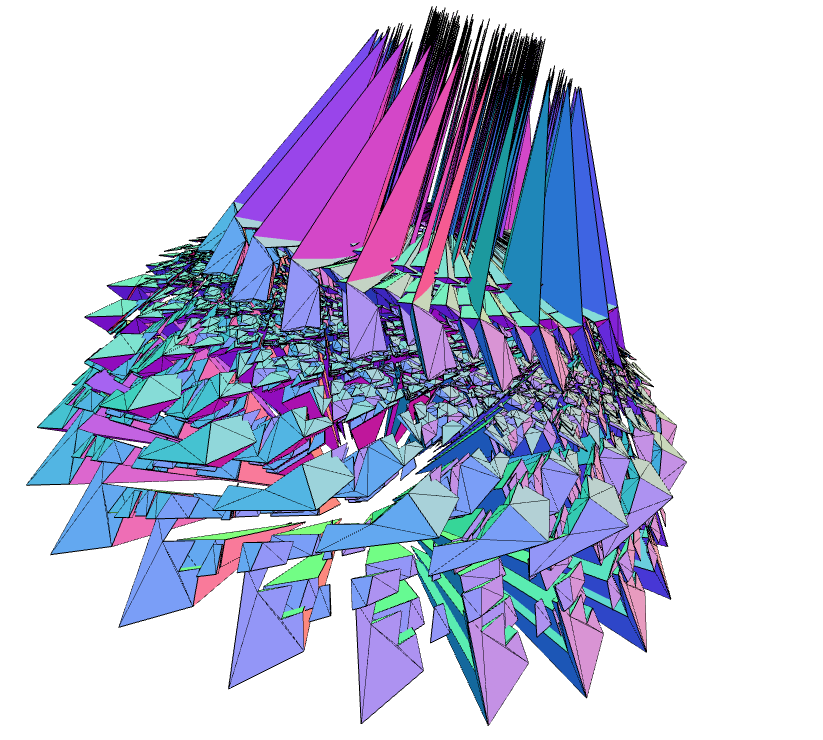}
\caption{$m=1$}
\end{subfigure}
\begin{subfigure}{0.3\textwidth}
\centering
\includegraphics[height=5cm]{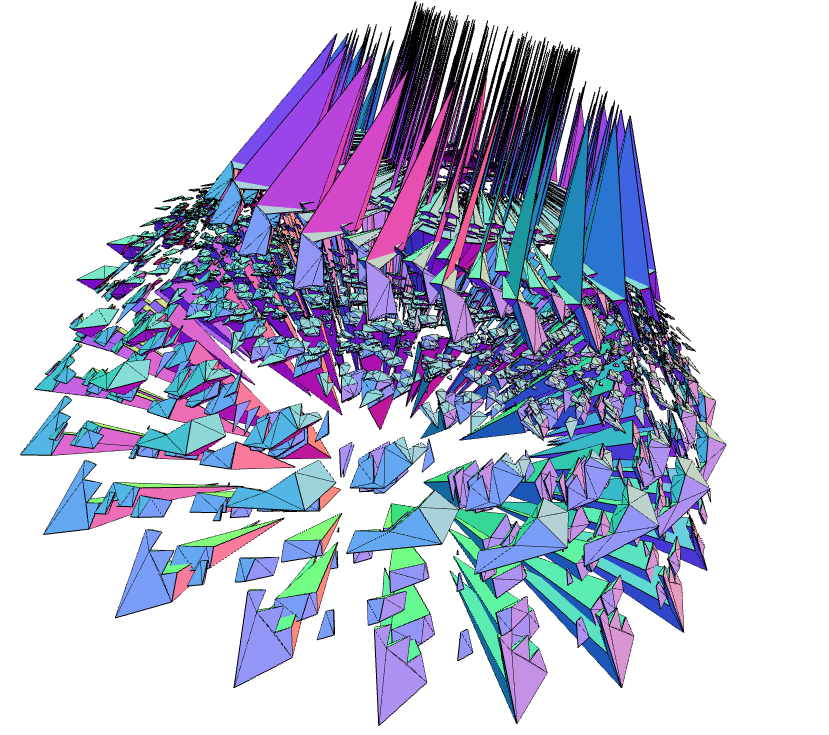}
\caption{$m=2$}
\end{subfigure}
\caption{Exploded views of the retriangulation of the simplicial surface $X_{23}$ with cyclic symmetry $C_{23}$ with different magnitudes $m$.}
\label{c23_explode}
\end{figure}

\section{Non-Manifold Parts} \label{sec:ramifications}
As seen in Figure \ref{fig:cube_examples} of the introduction, a common issue that might arise in a model is the existence of non-manifold parts. Removing these is crucial in many applications. This can be achieved by transforming the complex by a small margin to obtain a similar, but non-degenerate surface. Combinatorially, the process of remedying non-manifold parts of an embedded complex corresponds to turning it into a simplicial surface by changing the parts that do not fulfil conditions (i) and (ii) from Definition \ref{def:simplicial_surface}. We refer to this process as \emph{remedying} the complex. Example are shown in Figure \ref{fig:TouchingCubes_remedy} and Figure \ref{fig:remedy_nm_edges}. \\
In 3D printing, non-manifold parts cause problems or instabilities during the printing process, as a 3D printer can not print sets of zero measure. At the same time, they occur naturally, such as in the data set of embedded icosahedra of edge length one, see \cite{IcosahedraEdgeLength1}. To avoid unnecessary computations, it is reasonable to consider these after first fixing self-intersections and computing the outer hull (for example, Figure \ref{fig:nonmanifold_edge_surf} could be part of an inner chamber). Moreover, inferring non-manifold parts from our combinatorial data can only be done after computing a self-intersection free complex that is geometrically equivalent to the initial one, see Section \ref{sec:preliminaries}. Additionally, to guarantee correct information on how to modify edges and vertices, one needs to ensure outward orientation of normals beforehand. Thus, in the following, we assume a complex with embedding $(X,\phi)$ that is intersection-free and reduced to its outer hull. \\

\begin{figure}[H]
\centering
\begin{minipage}{0.49\textwidth}
   \begin{subfigure}{\textwidth}
   \centering
    \includegraphics[height=3cm]{images/TouchingCubes.png}
    \includegraphics[height=3cm]{images/TouchingCubesFixedSlice.png}
    \caption{}
    \label{fig:TouchingCubes_remedy}
   \end{subfigure} 
\end{minipage}
\begin{minipage}{0.49\textwidth}
   \begin{subfigure}{\textwidth}
   \centering
    \includegraphics[height=3cm]{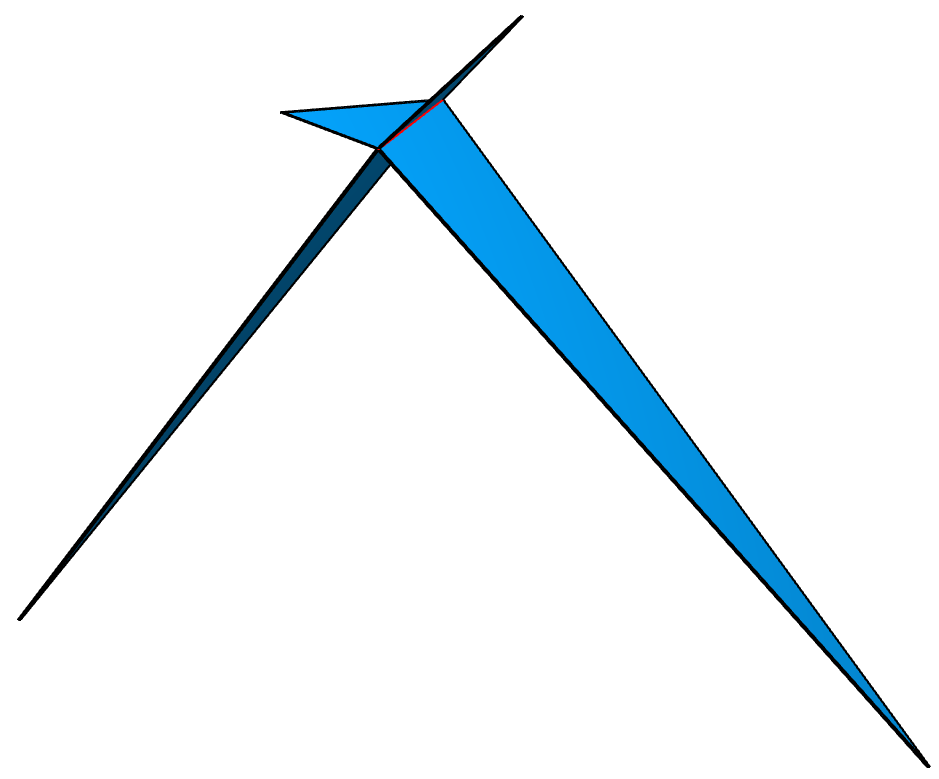}
    \includegraphics[height=3cm]{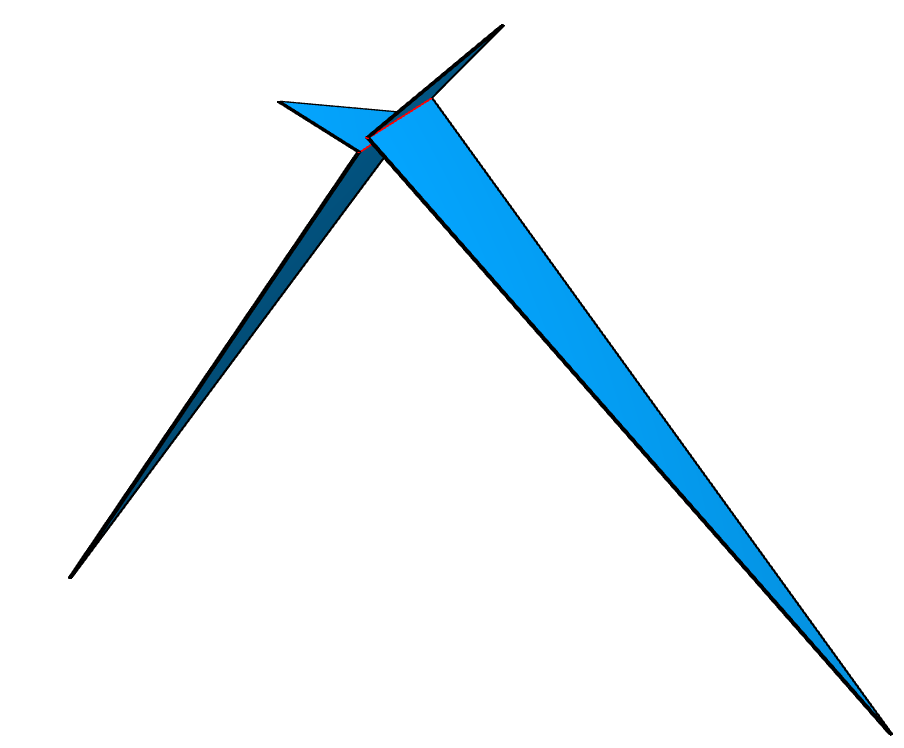}
    \caption{}
    \label{fig:fan_remedy}
   \end{subfigure} 
\end{minipage}
\caption{(a) Remedying of two adjoint cubes. On the left, the original complex is shown, while the complex on the right possesses only manifold edges. (b) Our remedy algorithm shown for a single fan around the red edge. This example is taken from icosahedron$_{3,2}$, which can also be seen in Figure \ref{fig:remedy_nm_edges} below.}
\label{fig:nonmanifold_cube_examples}
\end{figure}

\subsection*{Non-Manifold Edges}
We first clarify notation and describe of our approach. Following this, we illustrate the assignment of new coordinates to vertices.
\begin{definition}
    \label{def:nm_edge}
    \textit{Non-manifold edges} are edges which are incident to more than two triangles, as shown in Figure \ref{fig:nonmanifold_edge_surf}. We differentiate between three types:
    \begin{enumerate}[(i)]
        \item \textit{Inner} non-manifold edges have both of their vertices incident to another non-manifold edge;
        \item \textit{Outer} non-manifold edges are edges where only one vertex is incident to another non-manifold edge;
        \item \textit{Isolated} non-manifold edges are edges which are not incident to any other non-manifold edge.
    \end{enumerate}
\end{definition}

It is important to note that one can treat isolated non-manifold edges by subdividing the incident triangles, resulting in two outer non-manifold edges. We thus only discuss the first two cases. \\ 
Inner and outer non-manifold edges occur in paths, with outer non-manifold edges at the start and end, respectively, and inner ones in between. We thus proceed in a path-based framework: First, we gather all the non-manifold edges and then treat edge-paths of non-manifold edges iteratively. This is visualised in Figure \ref{fig:remedy_nm_edges}.
\begin{figure}[H]
\begin{subfigure}{0.32\textwidth}
\centering
\includegraphics[height=5cm]{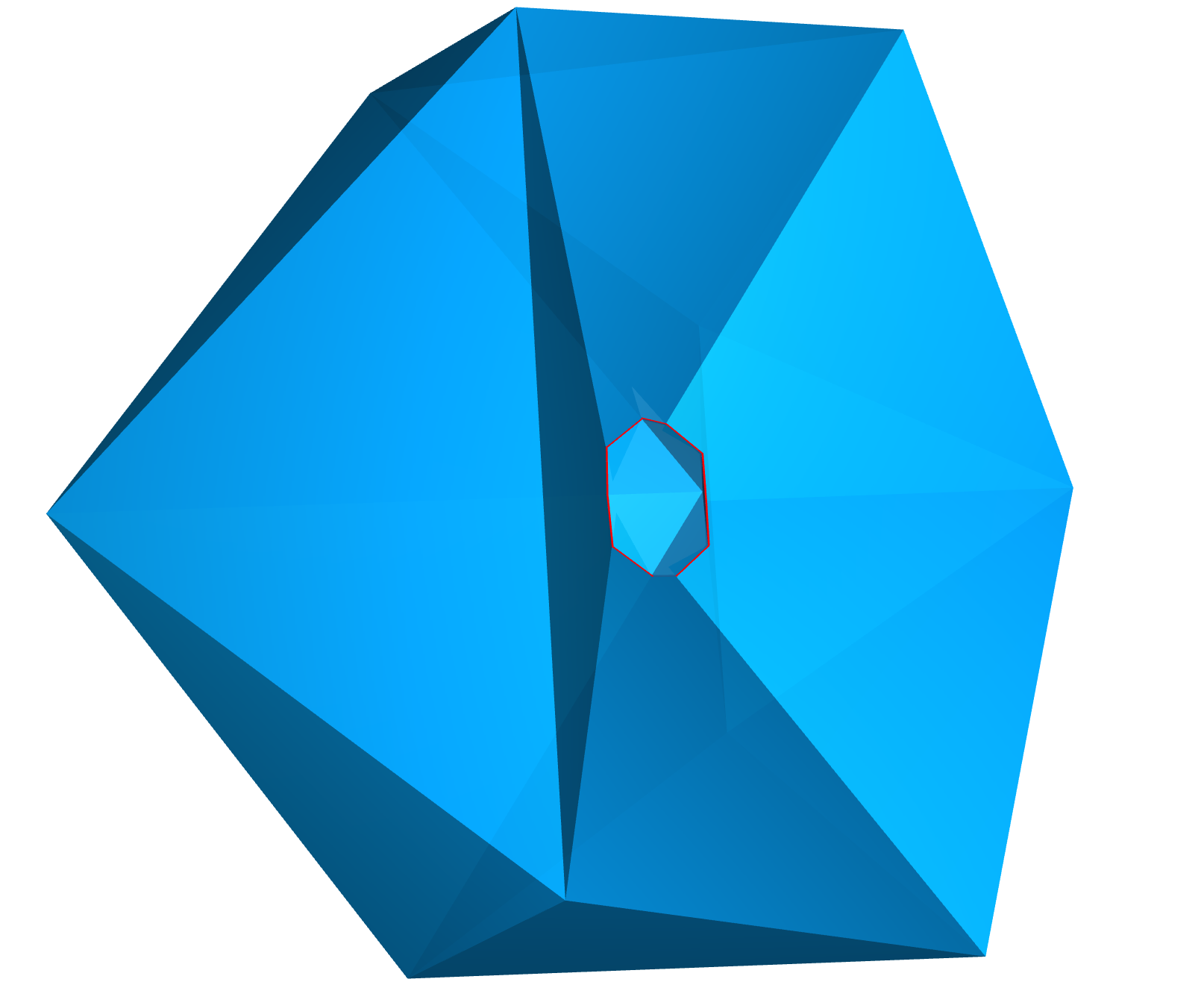}
\caption{}
\label{fig:nonmanifold_edge_surf}
\end{subfigure}
\begin{subfigure}{0.32\textwidth}
\centering
\includegraphics[height=5cm]{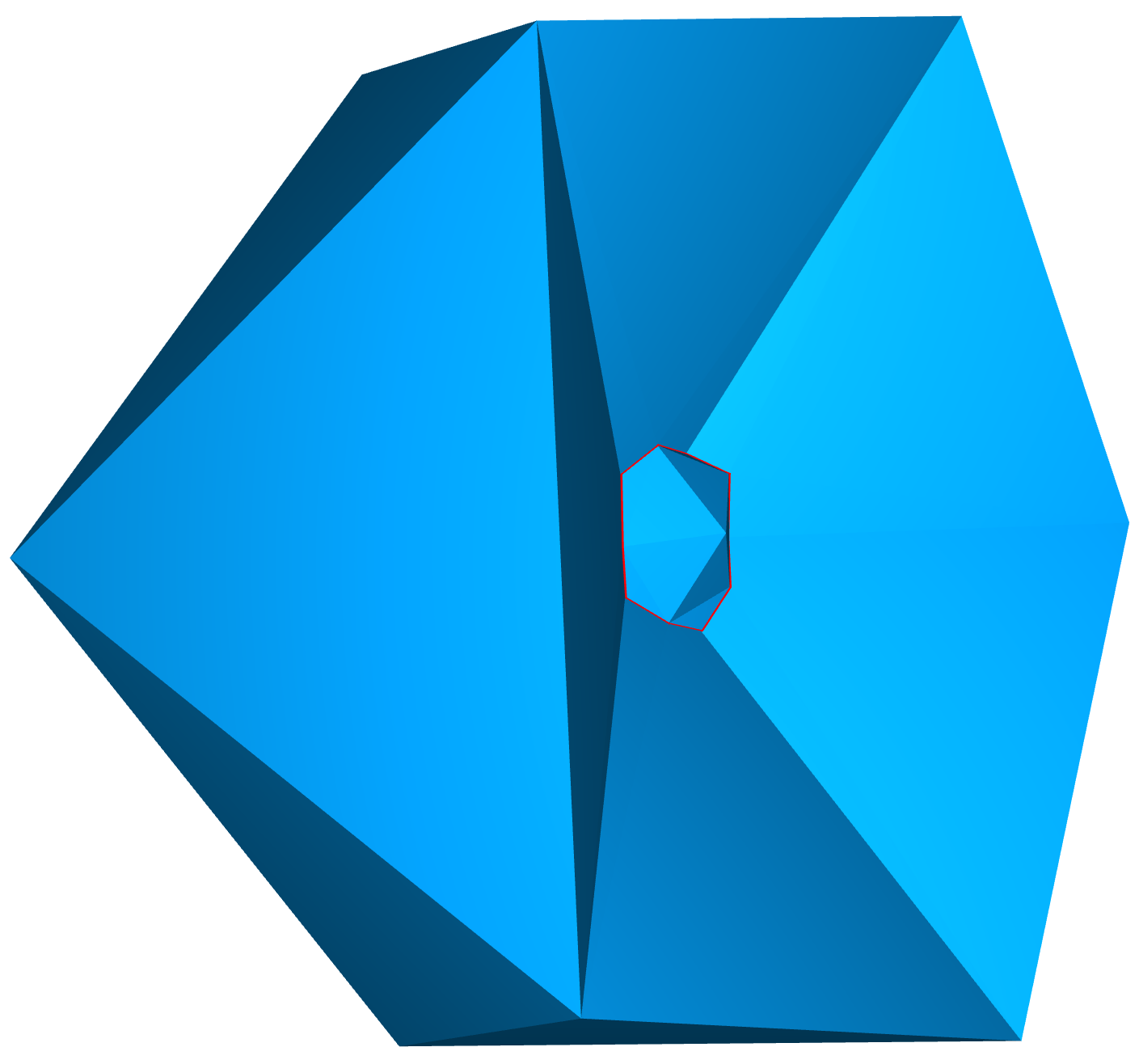}
\caption{}
\label{fig:remedied_nonmanifold_edge_surf}
\end{subfigure}
\begin{subfigure}{0.32\textwidth}
\centering
\includegraphics[height=5cm]{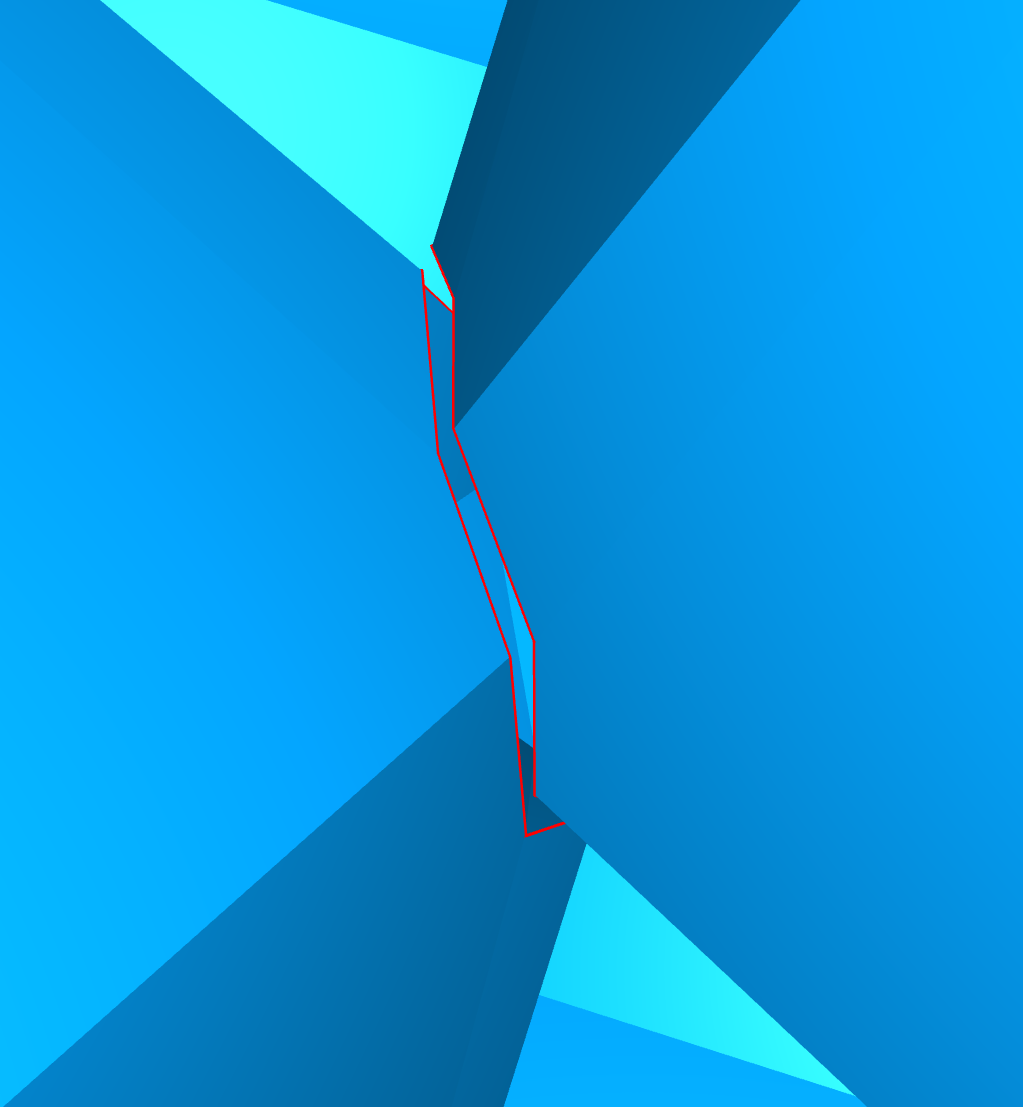}
\caption{}
\label{fig:remedied_nonmanifold_edge_surf_interior}
\end{subfigure}
\caption{(a) Non-manifold edges marked in red. (b) Non-manifold edges remedied. Note that the altered surface has only minimal differences to the original. (c) Interior view that shows split (prior non-manifold) edges.}
\label{fig:remedy_nm_edges}
\end{figure}

To remedy an inner non-manifold edge $e$, we split the vertices $v_1,v_2$ incident to $e$, resulting in new vertices $v^j_1,v^j_2$ for $j=1,2$. Additionally, we create another edge $e' =\{v^2_1,v^2_2\}$ and set $e =\{v^1_1,v^1_2\}$. \\
For an outer non-manifold edge, we only split the vertex incident to another non-manifold edge. Both cases can be seen in Figure \ref{fig:splitting_ram_edge}. The exact assignment of coordinates to the new vertices is described below. 

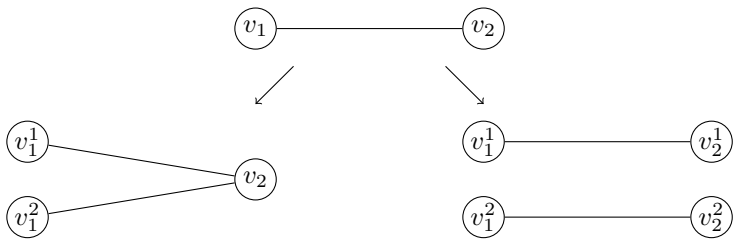
\begin{figure}[H]
    \centering
    \begin{tikzpicture}[node distance={10mm}, every node/.style={draw, minimum size=3.6ex,inner sep=0pt}, main/.style = {draw, circle}] 
    \node[main] at (-1,0) (1) {$v_1$}; 
    \node[main] at (2,0) (2) {$v_2$}; 
    \draw (1) -- (2);
    \node[main] at (2,-1.5) (3) {$v_1^1$}; 
    \node[main] (4) [below of=3]{$v_1^2$}; 
    \node[main] at (5,-1.5) (5) {$v_2^1$}; 
    \node[main] (6) [below of=5]{$v_2^2$}; 
    \draw (3) -- (5);
    \draw (4) -- (6);
    \node[main] at (-4,-1.5) (7) {$v_1^1$}; 
    \node[main] (8) [below of=7]{$v_1^2$}; 
    \node[main] at (-1,-2) (9) {$v_2$}; 
    \draw (7) -- (9);
    \draw (8) -- (9);
    \draw[->] (-0.5,-0.5) -- (-1,-1);
    \draw[->] (1.5,-0.5) -- (2,-1);
    \end{tikzpicture}
    \caption{Splitting non-manifold Edges based on type: outer (left) and inner (right).}
    \label{fig:splitting_ram_edge}
\end{figure}

Every path of non-manifold edges of length $\ell$ that is not a circle starts and ends at an outer non-manifold edge, and thus we only need to split $\ell -1$ vertices (one for each edge except the last). For circles, one then has to split $\ell$ vertices.\\
One also needs to consider \emph{junctions} of such paths, so vertices where more than two non-manifold edges meet. If, while iterating, we come upon a junction, we proceed via a depth first approach and split the remainder up into disjoint paths, since the edges are already fixed by splitting the common vertex. It is possible in certain configurations to produce a single non-manifold vertex after all non-manifold edges are corrected, which is why we deal with edges first and then consider vertices. Our algorithms also rely on this order to correctly infer information about the surface.

We now describe the coordinate assignment when splitting vertices. An important definition that we will need is the following.
\begin{definition}
    Let $(X,\phi)$ be an embedded simplicial complex, reduced to its outer hull, and $e \in X$ an edge. We call a pair of faces $(f_1,f_2)$ in $\mathrm{Fan}(e)$ an \textit{outward oriented butterfly} if, when rotating $f_1$ along its normal with rotation axis $e$, the next face in the fan that one encounters is $f_2$.
\end{definition}
For manifold edges $e$, the above is not very interesting: $\mathrm{Fan}(e)=\{f_1,f_2\}$ holds for an outward oriented butterfly $(f_1,f_2)$.
\begin{remark}
    Note that by our regularity assumptions, for an embedded simplicial complex that has been reduced to its outer hull, the fan of each non-manifold edge can be decomposed into $|\mathrm{Fan}(e)|/2$ outward oriented butterflies.
\end{remark}
To compute the new vertices, we shift them by some $\varepsilon >0$ into a suitable direction $s_e$, given as follows.
Choose an outward oriented butterfly $(f_1,f_2) \subset \mathrm{Fan}(e)$ and taking for both $f_1,f_2$ the vertex $w_i$ that is not incident to $e$, and set 
\begin{align*}
    s_e \coloneqq 1/2(w_1 - v_1) + 1/2 (w_2 - v_1).
\end{align*}
Then we replace $v$ in $f_1$ and $f_2$ by $\overline{v} := v + \varepsilon s_e$.\\
This procedure is iterated for all outward oriented butterflies of $\mathrm{Fan}(e)$, which thus remedies $e$.
\begin{remark}
    Note that the parameter $\varepsilon$ above can be adapted based on the considered application area. For example, for 3D printing, it can be chosen based on the nozzle width, which guarantees that the resulting surface is printed correctly.
\end{remark}

\subsection*{Non-Manifold Vertices}
As noted above, it is possible, in certain configurations, for a non-manifold vertex to appear after dealing with a junction of non-manifold edges. These can also be present as artefacts in 3D modelling, for example when adjoining surfaces. Note that it only makes sense to talk about non-manifold vertices here, not about non-manifold points, as these would constitute self-intersections which we removed beforehand.
\begin{definition}
\textit{Non-manifold vertices} are ones for which the \textit{umbrella condition} fails. Thus, the incident faces cannot be ordered in a connected face-edge path (see also the discussion proceeding Definition \ref{def:simplicial_surface}). So a neighbourhood of the vertex in the complex would fail to be connected if we removed the vertex.
\end{definition}
To remedy such a vertex $v$, we iterate over the family $C_\alpha$ of all \textit{local umbrellas} of $v$ (indexed by $\alpha$).
\begin{definition}[Local Umbrella]
    For a complex $S$ without non-manifold edges and vertex $v\in S$, the local umbrellas $C_\alpha$ of $v$ in $S$ are defined as follows. Take $F$ to be all the faces incident to $v$ and set $C_\alpha$ to be the equivalence classes of $\mathrm{Pot}(F)$ under the following relation
    \begin{align*}
        f_1 =_v f_2 \iff f_1 \text{ can be reached via an edge-face path from } f_2 \text{ by only using edges incident to }v.
    \end{align*}
    Thus, the local umbrellas $C_\alpha$ are maximal sets of faces of $v$ under the condition that all faces in the set are edge-connected via edges that are incident to $v$.
\end{definition}
\begin{remark}
    It is of importance to first remedy the non-manifold edges before considering vertices, which can be seen in the definition above. If there were still non-manifold edges, the local umbrella could also switch between different sides of the surface.
\end{remark}
For each local umbrella $C_\alpha$, we take for all $f\in C_\alpha$ their vertices $v^f_1$ and $v^f_2$ that do not equal $v$ and compute an average direction vector
\begin{align*}
    v_\alpha := \dfrac{1}{|C_\alpha|}\left(\sum_{f\in C_\alpha} \dfrac{1}{2}(v^f_1 - v_0)+\dfrac{1}{2}(v^f_2 - v_0)\right).
\end{align*}
Next, we replace $v$ by $v' := v_0 + v_\alpha\cdot \varepsilon$ in all the $f$ of our current $C_\alpha$, where $\varepsilon > 0$ is a small shift parameter.\\
We thus move $v$ slightly in the direction of the local umbrella, and by doing this for each of the local umbrellas, we obtain $m$-moved versions of $v$ that are pairwise distinct and manifold, where $\alpha = 1,\dots,m$ is the number of $v$-edge-connected components of the faces of $v$.

\section{Experiments on Self-Intersecting Icosahedra}

\begin{figure}[H]
\centering
\includegraphics[width=\textwidth]{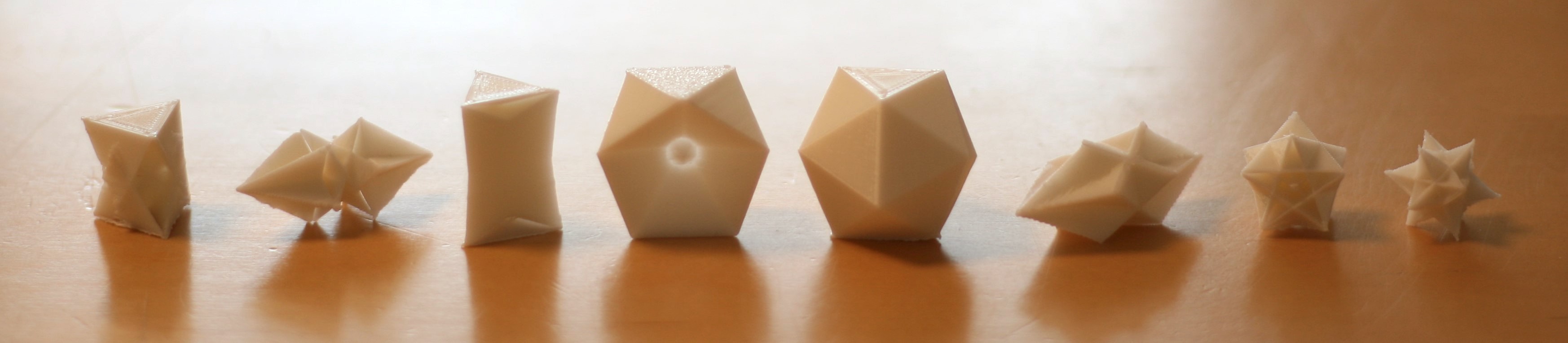}
\caption{3D printed icosahedra with constant edge length.}
\label{fig:printed_icos}
\end{figure}

In \cite{IcosahedraEdgeLength1} all $35$ symmetric embeddings of icosahedra with equilateral triangles of edge length $1$ are classified. This means that the underlying simplicial surface agrees but the embedding varies, see \cite{Icosahedra:online} for coordinates and visualizations of the resulting surfaces. From these $35$ embeddings, $33$ have self-intersections and thus are candidates for testing the algorithms presented in the previous sections. They are named as icosahedron$_{i,j}$, where $i\in \{1,\dots,14\}$ and $j\in\{1,\dots,k(i)\}$. The numbering comes from the classification based on $14$ formal Gram-matrices that give rise to different embeddings. Note that, for a given formal Gram-matrix, all corresponding embeddings have the same symmetry, see \cite{IcosahedraEdgeLength1}.

In Figure \ref{great_ico_explode}, we see one such an icosahedron, also known as the great icosahedron, with symmetry group isomorphic to the entire icosahedral group. Thus, we can use symmetric optimization as proposed in Section \ref{sec:symmetry}. Table \ref{tab:icosahedra_timing} in the Appendix shows that applying our symmetry algorithm leads to large speed up by a factor of $14.68$ when first retriangulating the surface and subsequently computing the outer hull. In Figure \ref{fig:printed_icos} we show a collection of printed versions of the icosahedra, beforehand processed by our algorithm.

The computations were performed on a MacBook Pro with Apple M1 Pro chip and 16GB memory, averaging over 10 run times. We observe a strong correlation between the group sizes and the average speedup when comparing the time needed with or without using the symmetry group.

\section{Conclusion and Outlook}
In Table \ref{tab:icosahedra_timing} of the Appendix, we can see the speedup for retriangulating and computing the outer hull of all $35$ symmetric icosahedra with edge length $1$. For an icosahedron with symmetry group size of 2, which corresponds to just a single symmetry in the complex, the average speedup in our data set is still $2.12$. Even if no symmetries are known a-priori, symmetry detection methods  \cite{NewSymmetryDetection,EfficientSymmetry} can be used to compute the symmetry group to allow the application of the presented symmetry-driven algorithm in Section \ref{sec:symmetry}. The symmetry optimisation method can be also applied to infinite structures if the face-orbits are finite. For instance, crystallographic groups give examples of doubly periodic symmetries with face-orbit representatives sitting inside a fundamental domain, see \cite{TopologicalInterlocking,TPMS} for examples of such surfaces.\\
The study of chambers of an embedded complex can be used to compute various geometric properties such as the volume and has the potential of designing interlocking puzzles \cite{InterlockingPuzzles} by subdividing a given complex into its chambers and introducing connectors. With this, we can also subdivide a given surface into chambers by introducing inner triangles leading to a complex inner structure.

The retriangulation method in this paper is motivated by producing a complex with a minimal number of triangles. The robustness of the underlying algorithm and the possibility of applying symmetries leads to a numerically robust and fast method. The results we obtained for embedded icosahedra and cyclic complexes show that a speedup factor corresponding to the number of non-trivial symmetries is possible without any loss of algorithmic stability.

For further research, optimisation of computing the outer hull of a self-intersecting complex can be studied.

For instance, one could consider finding alternative algorithms that compute the outer hull before considering self-intersections, since only the outer hull is relevant in many settings. Additionally, since we rely on the exploitation of the mathematical structure of simplicial complexes, we need suitable assumptions on initial data. One could, then, also consider addressing the problem of turning an arbitrarily degenerate input (like a non-closed surface) into one that can be treated by our framework.
Another interesting goal would be to remedy non-manifold parts in a way that preserves the symmetry group of the original complex. Then, algorithms for further processing could benefit from symmetry as well. In general, additional research is required to gain an understanding of how parameters need to be adjusted in order to guarantee a model that satisfies the same properties as the underlying geometric object.

\section{Acknowledgements}

This work was funded by the Deutsche Forschungsgemeinschaft (DFG, German Research Foundation) - SFB/TRR 280. Project-ID: 417002380. The authors thank Prof.~Alice C. Niemeyer for her valuable comments on earlier versions of this work.  The authors also thank Lukas Schnelle for support in generation of images.

\printbibliography
\section{Appendix}
\begin{table}[H]
\centering
\small
\caption{Timing of computing the outer hull of Icosahedra \cite{Icosahedra:online}: run times of algorithm on all 35 symmetric icosahedra with and without using their respective symmetry groups. The average speedup on our data set is approximately $3.17$.}
\begin{tabular}{c|c|c|c|c}
\hline
\textbf{Icosahedron} & \textbf{Group Size} & \textbf{Time with Group (ms)} & \textbf{Time without Group (ms)} & \textbf{Speedup} \\
\hline
Icosahedron$_{1,1}$ & 4 & 30.6 & 83.0 & 2.71 \\
\hline
Icosahedron$_{2,1}$ & 120 & 715.4 & 10502.6 & 14.68 \\
\hline
Icosahedron$_{2,2}$ & 120 & 13.6 & 67.4 & 4.96 \\
\hline
Icosahedron$_{3,1}$ & 20 & 75.0 & 385.6 & 5.14 \\
\hline
Icosahedron$_{3,2}$ & 20 & 280.8 & 116.0 & 3.80 \\
\hline
Icosahedron$_{4,1}$ & 12 & 35.8 & 80.9 & 2.31 \\
\hline
Icosahedron$_{4,2}$ & 12 & 84.9 & 424.1 & 2.26 \\
\hline
Icosahedron$_{5,1}$ & 4 & 33.2 & 144.7 & 5.00 \\
\hline
Icosahedron$_{5,2}$ & 4 & 116.0 & 622.1 & 4.36 \\
\hline
Icosahedron$_{6,1}$ & 20 & 39.2 & 140.6 & 5.36 \\
\hline
Icosahedron$_{7,1}$ & 12 & 1090.5 & 2536.7 & 2.33 \\
\hline
Icosahedron$_{7,2}$ & 12 & 488.8 & 1343.2 & 2.75 \\
\hline
Icosahedron$_{7,3}$ & 12 & 212.8 & 517.5 & 2.43 \\
\hline
Icosahedron$_{8,1}$ & 4 & 111.5 & 325.1 & 2.92 \\
\hline
Icosahedron$_{9,1}$ & 4 & 229.8 & 409.1 & 1.78 \\
\hline
Icosahedron$_{9,2}$ & 4 & 2074.3 & 4097.2 & 1.98 \\
\hline
Icosahedron$_{9,3}$ & 4 & 89.8 & 149.2 & 1.66 \\
\hline
Icosahedron$_{9,4}$ & 4 & 136.8 & 241.9 & 1.77 \\
\hline
Icosahedron$_{9,5}$ & 4 & 58.3 & 99.3 & 1.70 \\
\hline
Icosahedron$_{10,1}$ & 2 & 679.9 & 3024.7 & 4.45 \\
\hline
Icosahedron$_{10,2}$ & 2 & 21.2 & 66.9 & 3.16 \\
\hline
Icosahedron$_{11,1}$ & 10 & 99.1 & 248.1 & 2.50 \\
\hline
Icosahedron$_{11,2}$ & 10 & 42.6 & 111.6 & 2.62 \\
\hline
Icosahedron$_{12,1}$ & 6 & 137.9 & 213.0 & 1.54 \\
\hline
Icosahedron$_{12,2}$ & 6 & 109.7 & 163.8 & 1.49 \\
\hline
Icosahedron$_{12,3}$ & 6 & 145.9 & 240.3 & 1.65 \\
\hline
Icosahedron$_{12,4}$ & 6 & 223.2 & 380.5 & 1.70 \\
\hline
Icosahedron$_{13,1}$ & 2 & 352.8 & 567.4 & 1.61 \\
\hline
Icosahedron$_{13,2}$ & 2 & 751.4 & 1130.4 & 1.50 \\
\hline
Icosahedron$_{13,3}$ & 2 & 460.7 & 703.2 & 1.53 \\
\hline
Icosahedron$_{13,4}$ & 2 & 467.0 & 780.6 & 1.67 \\
\hline
Icosahedron$_{13,5}$ & 2 & 153.2 & 236.4 & 1.54 \\
\hline
Icosahedron$_{13,6}$ & 2 & 243.9 & 373.4 & 1.53 \\
\hline
Icosahedron$_{14,1}$ & 10 & 57.3 & 182.9 & 3.19 \\
\hline
\end{tabular}
\label{tab:icosahedra_timing}
\end{table}

\end{document}